\documentclass{sig-alternate}
\usepackage{booktabs}
\usepackage{graphicx,subfigure}
\usepackage{algorithm} 
\usepackage{algorithmic} 
\usepackage{multirow} 
\usepackage{amsmath}
\usepackage{xcolor}
\usepackage{graphicx}
\usepackage{subfigure}
\usepackage{amsmath,amssymb,amsfonts}
\newtheorem{theorem}{Theorem}
\newtheorem{definition}{Definition}
\newtheorem{lemma}{Lemma}

\begin{document}
%

\title{IMRank: Influence Maximization via Finding Self-Consistent Ranking}
%
%
%
%
%

\numberofauthors{1} 
%
\author{
%
%
\alignauthor
Suqi Cheng, Hua-Wei Shen, Junming Huang, Wei Chen, Xue-Qi Cheng\\
      \affaddr{Institute of Computing Technology, Chinese Academy of Sciences, Beijing, China}\\
       \email{\{chengsuqi, shenhuawei, huangjunming, chenwei2012, cxq\}@ict.ac.cn}
}

\maketitle
%
%

\begin{abstract}

Influence maximization, fundamental for word-of-mouth marketing and viral marketing, aims to find a set of seed nodes maximizing influence spread on social network. Early methods mainly fall into two paradigms with certain benefits and drawbacks: (1)Greedy algorithms, selecting seed nodes one by one, give a guaranteed accuracy relying on the accurate approximation of influence spread with high computational cost; (2)Heuristic algorithms, estimating influence spread using efficient heuristics, have low computational cost but unstable accuracy.

We first point out that greedy algorithms are essentially finding a self-consistent ranking, where nodes' ranks are consistent with their ranking-based marginal influence spread. This insight motivates us to develop an iterative ranking framework, i.e., IMRank, to efficiently solve influence maximization problem under independent cascade model. Starting from an initial ranking, e.g., one obtained from efficient heuristic algorithm, IMRank finds a self-consistent ranking by reordering nodes iteratively in terms of their ranking-based marginal influence spread computed according to current ranking. We also prove that IMRank definitely converges to a self-consistent ranking starting from any initial ranking. Furthermore, within this framework, a last-to-first allocating strategy and a generalization of this strategy are proposed to improve the efficiency of estimating ranking-based marginal influence spread for a given ranking. In this way, IMRank achieves both remarkable efficiency and high accuracy by leveraging simultaneously the benefits of greedy algorithms and heuristic algorithms. As demonstrated by extensive experiments on large scale real-world social networks, IMRank always achieves high accuracy comparable to greedy algorithms, with computational cost reduced dramatically, even about $10-100$ times faster than other scalable heuristics.

%

\end{abstract}


\section{Introduction}

The prosperity of online social networks and social media invokes a new wave of research on social influence analysis~\cite{Tang2009,Huang2012}.
Finding influential individuals is important for many applications such as expert finding, online advertising and marketing. Therefore, influence maximization is identified as a fundamental problem for word-of-mouth marketing and viral marketing in the area of online marketing. It aims to find a fixed-size set of seed nodes in social network to maximize their \emph{influence spread}, i.e., the expected number of activated nodes triggered by the seed nodes. Ever since being formalized by Kempe et al.~\cite{Kempe2003}, influence maximization problem has attracted much research attention from various fields, including social network analysis, data mining and marketing.

Early methods for influence maximization mainly use greedy framework, selecting one by one the node with the largest marginal influence spread. With calculating influence spread  accurately, the greedy framework is proved to provide a $(1-1/e)$ approximation to the optimal solution of influence maximization~\cite{Kempe2003}, guaranteed by the submodularity and monotonicity properties of influence spread as a function of seed node set. These methods roughly fall into two paradigms: greedy algorithms~\cite{Kempe2003,Leskovec2007,Chen2009,Goyal2011,Cheng2013} and heuristic algorithms~\cite{Kimura2010,Chen2010,Wang2010,Jung2012}. Greedy algorithms provide a $(1-1/e-\epsilon)$ approximation by approximating influence spread through Monte Carlo simulation. However, they have high computation cost because the calculation of marginal influence spread invokes estimating the influence spread of nodes from scratch, using time-consuming Monte Carlo simulation. The latter, in contrast, resorts to estimate the influence spread via efficient heuristic methods. The scalability of these heuristics generally outperforms the greedy algorithms by several orders of magnitude.  Yet, their high scalability is gained with the pain of unguaranteed accuracy and unreliable performance on various scenarios. To the best of our knowledge, we lack an efficient and accurate algorithm of influence maximization for applications to large scale social networks in real world.

In this paper, we propose an efficient and accurate algorithm to solve influence maximization problem under independent cascade model. This algorithm is motivated by the key insight that greedy algorithms are essentially finding a self-consistent ranking, where nodes' ranks are consistent with their ranking-based marginal influence spread. We prove that such self-consistent ranking can be obtained directly using an iterative ranking framework, i.e., IMRank, proposed in this paper. Starting from an initial ranking, e.g., one obtained from efficient heuristic algorithm, IMRank efficiently finds a self-consistent ranking by reordering nodes iteratively in terms of their ranking-based marginal influence spread computed according to current ranking. Different from greedy algorithms computing ranking-based marginal influence spread from scratch, IMRank conducts the computation of ranking-based marginal influence spread via an efficient last-to-first allocating strategy. As a result, IMRank achieves both high efficiency and high accuracy by leveraging simultaneously the benefits of greedy algorithms and heuristic algorithms. 

To evaluate the performance of IMRank, we conduct extensive experiments on large-scale social networks with hundreds of thousands of edges to millions of edges. Experimental results demonstrate that IMRank achieves high accuracy comparable to greedy algorithms with computational cost reduced dramatically. 

Our main contributions are summarized as follows:

\begin{itemize}
  \item We propose a novel framework IMRank, which unifies the estimation of marginal influence spread and the selection of seed nodes. IMRank achieves both remarkable efficiency and high accuracy by exploiting the interplay between the calculation of ranking-based marginal influence spread and the ranking of nodes.

  \item We prove that IMRank, starting from any initial ranking, definitely converges to a self-consistent ranking in a finite number of steps. This indicates that IMRank is efficient at solving the influence maximization problem via finding the final self-consistent ranking.

  \item We design an efficient last-to-first allocating strategy to approximately estimate the ranking-based marginal influence spread of nodes for a given ranking, further improving the efficiency of IMRank.

  \item We conduct extensive experiments on several real-world networks under different types of the independent cascade model. Through comparing two instances of IMRank with both greedy algorithm and existing state-of-the-art heuristics, we show that IMRank always achieves comparable accuracy to the greedy algorithm while runs $10-100$ times faster than other heuristics with better accuracy.
\end{itemize}

\section{Related Work}\label{section:relatedwork}

\begin{table*}[t]
\centering \caption{Notations.} \label{table:natations}
\begin{tabular}{|c|p{400pt}|}
\hline
Notation            &   Description \\ \hline
$v_i$               &   a node with index $i$ \\
$r_i$               &   the index of node with rank $i$ with respect to a given ranking $r$   \\  
$S=\{v_1,v_2,\dots,v_n\}$   &   a set of nodes   \\
$I(S)$              &   expected number of nodes eventually activated by set $S$ \\
$M(v|S)$            &   marginal influence spread by adding node $v$ into a seed set $S$ \\
$M_r(v)$            &   short for $M(v|\{v_1,v_2,\dots,v_{i-1}\})$, where $\{v_1,v_2,\dots,v_{i-1}\}$ is the set of nodes ranked higher than $v$ in a given ranking $r$ \\
$p(v_i|\{v_1,v_2,\dots,v_{i-1}\})$  &   probability that $v_i$ is activated given that a collection of nodes $\{v_1,v_2,\dots,v_{i-1}\}$ are already activated   \\
$\eta_r(v_i,v_j)$   &   influence score that node $v_i$ sends to node $v_j$ with respect to a given ranking $r$ \\
$d(v_j,v_i)$      &     a simple path starting from $v_j$ and ending at $v_i$, i.e., $\{w_1=v_j,w_2,\dots,w_n=v_i\}$ \\
$d_r(v_j,v_i)$      &   influence path, which is a simple path where $v_j$ is the only node ranked higher than $v_i$ on the path\\
$\rho_r(v_i,v_j)$   &   probability that $v_i$ is activated by $v_j$ through all influence paths, with respect to a given ranking $r$ \\
$l$                 &   maximal length of all influence paths to account into \\
\hline
\end{tabular}
\end{table*}

Influence maximization problem was first studied by Domingos and Richardson from algorithmic perspective~\cite{Domingos2001,Richardson2002}. Kempe et al. then formulated it as a combinatorial optimization problem of finding a set of seed nodes with maximum influence spread~\cite{Kempe2003}. They proved that this problem is NP-hard and proposed a greedy algorithm which can guarantee a $(1-1/e-\epsilon)$ approximation ratio. Here, $\epsilon$ is caused by the inaccurate estimation of influence spread~\cite{Chen2010}~\cite{Chen2010b}. The biggest problem suffered by Kempe's greedy algorithm is its low scalability, limiting it to social networks with small or moderate size.

Many efforts have been made to improve the scalability of Kempe's greedy algorithm for influence maximization. ``cost-effective lazy forward'' (CELF) optimization strategy~\cite{Leskovec2007} and CELF++~\cite{Goyal2011} are proposed to reduce the times of influence spread estimation in Kempe's greedy algorithm by exploiting the submodularity property of influence spread function. To reduce the number of Monte Carlo simulations, Chen et al. proposed NewGreedy algorithm and MixedGreedy algorithm in~\cite{Chen2009}. The NewGreedy algorithm reusing the results of Monte Carlo simulations in the same iteration to calculate marginal influence spread for all candidate nodes. Yet, it increases the computational cost for a single Monte Carlo simulation because the simulation is now conducted globally rather than locally as done in Kempe's greedy algorithm. As a remedy, the MixedGreedy algorithm was developed, integrating the CELF strategy into the NewGreedy algorithm. Sheldon et al.~\cite{Sheldon2010}  proposed a sample average approximation approach from stochastic optimization for maximizing the spread of cascades under budget restriction. Cheng et al. proposed a static greedy algorithm~\cite{Cheng2013}, reducing the number of Monte-Carlo simulations through strictly guaranteeing the submodularity and monotonicity properties of influence spread function. Although these improvements can speedup the original greedy algorithm in several orders of magnitude, scalability is still a big challenge for greedy algorithms because the guaranteed accuracy of these algorithms relies on a huge number of Monte Carlo simulations.

Heuristic algorithms, in contrast, mainly reduce the complexity of Kempe's greedy algorithm through computing influence spread heuristically. DegreeDiscount, designed for uniform independent cascade model, only computes direct influence~\cite{Chen2009}. Community-based greedy algorithm conducted Monte Carlo simulation within each community rather than on the whole network~\cite{Wang2010}.  SPM/SP1M algorithms~\cite{Kimura2010} estimated influence spread according to shortest paths, while PMIA algorithm~\cite{Chen2010} used maximum influence paths. SP1N algorithm employed the concept of Shapley value from the cooperative game theory~\cite{Narayanam2011}. IRIE algorithm efficiently estimated marginal influence spread through an iterative method. Besides the above heuristics using greedy approach, Jiang et al. proposed a simulated annealing approach with several heuristics~\cite{Jiang2011}, and Mathioudakis et al. suggested to speed up influence maximization using a simplified influence network~\cite{Mathioudakis2011}. However, these heuristics cannot give rise to guaranteed accuracy and their performance is unstable on different networks and diffusion models.

Taken together, in existing algorithms for influence maximization, the estimation of influence spread and the ranking of nodes are studied separately. On one hand, without leveraging the ranking of nodes, greedy algorithms estimate the influence spread of nodes from scratch, causing high computational cost. On the other hand, lacking a reliable estimation of influence spread, heuristic algorithms have no guaranteed accuracy. Hence, in this paper, we improve the state-of-the-art solution of influence maximization problem by exploiting the interplay between marginal influence spread and the ranking of nodes.

\section{Self-consistent ranking}\label{section:consistentrank}

For influence maximization on a social network $G=(V,E)$, \emph{influence spread function} $I(S)$ of a node set $S\subseteq V$ is defined as the expected number of nodes in $G$ eventually activated by $S$ under certain diffusion model. The function $I(\cdot)$ is nonnegative, monotone, and submodular, satisfying
\begin{itemize}
  \item Nonnegative: $I(S)\geq 0$;
  \item Monotone: $I(S) \leq I(T)$, if $S\subseteq T \subseteq V$;
  \item Submodular: $I(S\cup \{v\})-I(S)\geq I(T\cup \{v\})-I(T)$, for all $v\in V$ and $S\subseteq T \subseteq V$.
\end{itemize}
These properties guarantee that a fair approximation to the optimal solution of influence maximization can be obtained by greedy algorithms, iteratively selecting the node with maximum marginal influence spread as seed node.

\begin{definition}{\textbf{Marginal influence spread:}}
Given a node set $S\subseteq V$ and a node $v\in V$, the marginal influence spread of $v$ upon $S$ is defined as $M(v|S) = I(S \cup \{v\})-I(S)$.
\end{definition}

However, the influence spread function is not extensive, i.e., $I(S\cup\{v\})\neq I(S)+I(\{v\})$ if $v\notin S$, since the nodes activated by $S$ may overlap with the nodes activated by $v$. Therefore, one has to compute the marginal influence spread by computing both $I(S)$ and $I(S\cup\{v\})$ from scratch, resulting in huge computation cost. To remedy this problem, we further analyze the property of the set of seed nodes obtained by greedy algorithms. Indeed, greedy algorithms implicitly give a ranking of nodes, where nodes are ranked in decreasing order of their marginal influence spread. Meanwhile, their marginal influence spread are computed based on their ranks in the implicit ranking. Hence, greedy algorithms obtain a \emph{self-consistent ranking} of nodes.

Before formally defining self-consistent ranking, we first introduce several related notations for clarity. Without loss of generality, we index all the nodes into $\{v_1,v_2,\cdots,v_n\}$ where $n=|V|$. A ranking of nodes, determined by a permutation $(r_1,r_2,\cdots,r_n)$ with $r_i\in \{1,2,\cdots,n\}$ denoting the index of node with rank $i$, is denoted as $r=\{v_{r_1},v_{r_2},\cdots,v_{r_n}\}$. With these notations, for convenience, we now define the ranking-based marginal influence spread of node with respect to a ranking $r$ as $M_r(v_{r_i}) = M(v_{r_{i}}|\{v_{r_1},v_{r_2},\cdots,v_{r_{i-1}}\})$. In addition, for clarity, Table~\ref{table:natations} lists all the important notations used in this paper.

\begin{definition}{\textbf{Self-consistent ranking:}}
A ranking $r$ is a self-consistent ranking iff $M_r(v_{r_i}) \geq M_r(v_{r_j}), \forall 1\leq i<j\leq n$.
\end{definition}

For the set of seed nodes obtained by greedy algorithms, there exists an interplay between the ranks of nodes and their marginal influence spread. On one hand, these nodes are ranked in descending order of their marginal influence spread. On the other hand, the marginal influence spread of nodes is calculated with respect to the ranks of nodes. Indeed, the set of seed nodes obtained by greedy algorithms forms a self-consistent ranking.

\begin{theorem}\label{lemma:self-consistentrank}
\textbf{Greedy algorithms for influence maximization gives a self-consistent ranking.}
\end{theorem}

\begin{proof}
Greedy algorithms iteratively select the node with maximum marginal influence spread as seed node. With a ranking $r$ denoting the order seed nodes are selected, we have $M(v_{r_i}|\{v_{r_1},v_{r_2},\cdots,v_{r_{i-1}}\}) \geq M(v_{r_j}|\{v_{r_1},v_{r_2},\dots,\\v_{r_{i-1}}\})$, for $i<j$. In addition,
the submodularity of influence spread function implies that $M(v_{r_j}|\{v_{r_1},v_{r_2},\cdots,v_{r_{i-1}}\}) \\\geq M(v_{r_j}|\{v_{r_1},v_{r_2},\cdots,v_{r_{j-1}}\})$. Using transitivity, we complete the proof with $M_r(v_{r_i}) = M(v_{r_i}|\{v_{r_1},v_{r_2},\cdots,v_{r_{i-1}}\})\geq M(v_{r_j}|\{v_{r_1},v_{r_2},\cdots,v_{r_{j-1}}\}) = M_r(v_{r_{j}})$.
\end{proof}

For a given social network, however, there are multiple self-consistent rankings besides the one obtained by greedy algorithms. Hence it is critical to develop effective algorithms to achieve a desired self-consistent ranking which is either the very ranking obtained by greedy algorithms or comparable to it from the point of influence maximization.

\section{IMRank}\label{section:IMRank}


In this section, we develop an efficient iterative framework IMRank to solve the influence maximization problem through finding a desired self-consistent ranking. IMRank distinguishes itself from greedy algorithms in one key point: in each iteration, IMRank efficiently estimates the marginal influence spread of all nodes based on the current ranking, while greedy algorithm compute the marginal influence spread from scratch with high computational cost. 

\subsection{IMRank: iterative framework}\label{subsection:descriptionofIMRank}


IMRank aims to find a self-consistent ranking from any initial ranking. It achieves the goal by iteratively adjusting current ranking as follows:
\begin{itemize}
  \item Compute the ranking-based marginal influence spread of all nodes $M_r$ with respect to the current ranking $r$;
  \item Obtain a new ranking by sorting all nodes according to $M_r$.
\end{itemize}
This iterative process is formally described in Algorithm~\ref{algorithm:IMRank}. It definitely converges to a self-consistent ranking, starting from any initial ranking (see Section~\ref{subsection:convergence} for proof). Intuitively, IMRank iteratively promotes influential nodes to top positions in the ranking, always increasing the influence spread of top-$k$ nodes during the process until it converges to a self-consistent ranking. Indeed, different initial rankings could make IMRank converge to different self-consistent rankings. We leave the discussion about initial ranking to Section~\ref{subsection:initialranking}.

\begin{algorithm}[t]
\caption{IMRank $(r)$}\label{algorithm:IMRank}
\begin{algorithmic}[1]
    \STATE $r^{(0)}=r$
    \STATE $t \leftarrow 0$
    \REPEAT
        \STATE $t \leftarrow t+1$
        \STATE Calculate $M_{r^{(t)}}$ with respect to the ranking $r$
        \STATE Generate a new ranking $r^{(t)}$ by sorting nodes in decreasing order according to $M_{r^{(t)}}$
    \UNTIL $r^{(t)}=r^{(t-1)}$
    \STATE output the self-consistent ranking $r^{(t)}$
\end{algorithmic}
\end{algorithm}
\subsection{Calculate ranking-based marginal influence spread}\label{subsection:calculate M}

%


The core step in IMRank is the calculation of ranking-based marginal influence spread. One straightforward way is to directly compute $M_r(v_{r_i})=M(v_{r_i}|\{v_{r_1},v_{r_2},\cdots,v_{r_{i-1}}\})$ using Monte Carlo simulation, as done by greedy algorithms. However, prohibitively high computational cost makes it impractical for IMRank.
To combat this problem, we propose a \emph{Last-to-First Allocating} (LFA) strategy to efficiently estimate $M_r$, leveraging the intrinsic interdependence between ranking and ranking-based marginal influence spread. We develop the LFA strategy under the widely-adopted independent cascade model~\cite{Kempe2003}. For the independent cascade model, when a node $u$ is activated, it has one chance to independently activate its neighboring nodes with a propagation probability $p(u,v)$ if $v$ has not been activated yet. Each node can be activated for only once.

The LFA strategy is based on the following fact: by definition, the ranking-based marginal influence spread $M_r(v)$ is equal to the expected number of nodes activated by $v$, given that when all nodes ranked higher than it have finished the propagation of their influence. This implies two basic rules under the calculation of $M_r(v)$:
\begin{enumerate}
  \item Each node can only be activated by nodes ranked higher than it in the given ranking;
  \item When a node could be activated by multiple nodes, higher-ranked node has higher priority to activate it.
\end{enumerate}

Following the two basic rules, the LFA strategy is described as follows:
\begin{itemize}
  \item Given a ranking $r$, the initial value of $M_r(v_{r_i})$ of each node is set to be $1$, satisfying the fact that the sum of $M_r(v_{r_i})$ over all nodes is equal to the number of nodes, since each node can only be activated once.
  \item Scanning the ranking from the last node to the top one, a fraction of $M_r(v_{r_i})$ is delivered to the nodes ranked higher than $v_{r_i}$, reflecting the first rule;
  \item The delivered influence score of $M_r(v_{r_i})$ is allocated among the nodes $v_j (j<i)$ in terms of their ranks, reflecting the second rule.

      Specifically, with $\eta(v_{r_j},v_{r_i})$ denoting the fraction of influence score delivered to node $v_{r_j}$ from node $v_{r_i}$, we have $\eta(v_{r_j},v_{r_i})=$
\begin{equation}\label{equation:allocation}
\left\{
\begin{aligned}
   & M_r(v_{r_i})p(v_{r_j}, v_{r_i})\prod_{k: 1\leq k<j}{\big(1-p(v_{r_k},v_{r_i})\big)}, j<i, \\
   & 0, \hspace{1.9in} \mbox{otherwise}.
\end{aligned}
\right.
\end{equation}
where $p(v_{r_j}, v_{r_i})$ is the propagation probability that node $v_{r_j}$ directly activates node $v_{r_i}$, known as \textit{a priori} for independent cascade model.
\end{itemize}
The calculation of the ranking-based marginal influence spread $M_r$ is completed after all nodes are scanned. The LFA strategy is formally depicted in Algorithm~\ref{algorithm:rankbasedinfluence}.
\begin{algorithm}[t]
\caption{Calculate $\mathbf{M_r}$($r$)}\label{algorithm:rankbasedinfluence}
\begin{algorithmic}[1]
    \FOR{$i=1$ to $n$}
      \STATE $M_r(v_{r_i}) \leftarrow 1$
    \ENDFOR
    \FOR{$i=n$ to $2$}
       \FOR{$j=1$ to $i$}
            \STATE $M_r(v_{r_j})$ $\leftarrow $ $M_r(v_{r_j})$ + $p(v_{r_j},v_{r_i}) \times M_r(v_{r_i})$
            \STATE $M_r(v_{r_i})$ $\leftarrow $ $\big(1-p(v_{r_j},v_{r_i})\big) \times M_r(v_{r_i})$
       \ENDFOR
    \ENDFOR
   \STATE output $M_r$
\end{algorithmic}
\end{algorithm}


Now we use an example to illustrate the LFA strategy. In Figure~\ref{fig:last-to-first-strategy}, $v_k$ denotes the node with rank $k$ for convenience, and $p_{i,j}$ is the propagation probability along edge $\langle v_i,v_j\rangle$. Here, the ranking is simply $r=\{1,2,3,4,5\}$. Solid lines represent the edges where influence could propagate, while dashed lines depict the edges where influence score is delivered when nodes are scanned. The lack of dashed line from node $v_3$ to node $v_2$ reflects that node $v_2$ is ranked higher than node $v_3$. For this case, the LFA strategy computes the ranking-based marginal influence spread as follows:
\begin{enumerate}
  \item Initially, $M_r(v_i)=1$ $1\leq i\leq 5$.
  \item Node $v_5$ is then scanned as the last node in the ranking. According to Equation(~\ref{equation:allocation}), $v_5$ delivers $p_{3,5}M_r(v_5)=p_{3,5}$ to $v_3$ and $p_{4,5}(1-p_{3,5})$ to $v_4$ respectively. Accordingly, $M_r(v_5)$ becomes $(1-p_{4,5})(1-p_{3,5})$.
  \item Then node $v_4$ is scanned. Since $M_r(v_4)$ is now $1+p_{4,5}(1-p_{3,5})$, $v_4$ delivers $p_{2,4}+p_{2,4}p_{4,5}(1-p_{3,5})$ to $v_2$. Note that the second item characterizes the influence of $v_2$ to $v_5$ through the path $\langle v_2,v_4,v_5\rangle$, reflecting that the LFA strategy could effectively capture the indirect influence among nodes. After $v_4$ is scanned, the final value of $M_r(v_4)$ is $(1-p_{2,4})(1+p_{4,5}(1-p_{3,5}))$.
  \item When node $v_3$ is scanned. it delivers $p_{1,3}(1+p_{3,5})$ to node $v_1$, with $(1-p_{1,3})(1+p_{3,5})$ remained.
  \item Finally, node $v_2$ is scanned. After $v_2$ is scanned, the final scores of $M_r(v_2)$ and $M_r(v_1)$ are $(1-p_{1,2})(1+p_{2,4}+p_{2,4}p_{4,5}(1-p_{3,5}))$ and $1+p_{1,2}(1+p_{2,4}+p_{2,4}p_{4,5}(1-p_{3,5}))+p_{1,3}(1+p_{3,5})$ respectively. The term $p_{1,2}p_{2,4}p_{4,5}(1-p_{3,5})$ in $M_r(v_1)$ captures the indirect influence from $v_1$ to $v_5$ through the path $\langle v_1,v_2,v_4,v_5\rangle$, indicating that the LFA strategy does collect influence with multiple intermediate nodes on the path. Note that it is not necessary to scan node $v_1$ since it does not delivery influence to other nodes
\end{enumerate}

The above illustration tells us that the LFA strategy efficiently calculates the ranking-based marginal influence spread for all nodes, scanning each node only once. Meanwhile, with indirect influence propagation being effectively captured, the LFA strategy provides a good delegate to calculate ranking-based marginal influence spread. We show the numerical results of the LFA strategy and 20,000 Monte Carlo simulations in the case of setting $p_{u,v}=0.2$ for all edges as done in uniform independent cascade model. As shown in Table~\ref{table:mrnew}, our strategy offers very close results to the time-consuming Monte Carlo simulations.

Finally, we summarize the LFA strategy by explaining why it works remarkably. First, it achieves its high efficiency by exploiting the interdependence between ranking and ranking-based marginal influence spread, avoiding the adoption of Monte Carlo simulations done in greedy algorithms. Second, it employs the intermediate nodes as delegates, in a last-to-first manner, to capture both direct and indirect influence propagation among nodes. In this way, ranking-based marginal influence spread could be efficiently calculated via scanning all nodes only once. In addition, we want to spell out that the LFA strategy only offers one effective approximation rather than exact calculation of influence spread. This is partly caused by the restriction that influence could only propagate from higher-ranked nodes to lower-ranked nodes. In Section~\ref{section:advance}, we will further improve the LFA strategy via relaxing this restriction.

\begin{figure}[t]
\centering
\includegraphics[width=0.4 \linewidth]{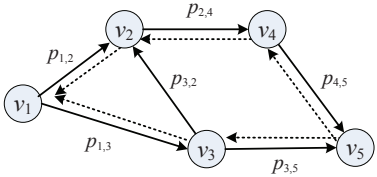}
\caption{Illustration of the LFA strategy.}
\label{fig:last-to-first-strategy}
\end{figure}


\begin{table}
\centering \caption{Estimation on ranking-based marginal influence spread. MC indicates Monte Carlo simulation, and LAF indicates the LAF strategy.}\label{table:mrnew}
\begin{tabular}{llllll}
\toprule
     & $v_1$  & $v_2$ & $v_3$ & $v_4$ & $v_5$\\
\midrule
MC  & 1.29846 &  1.38800 & 0.77941 & 0.89406  & 0.64007  \\
LAF            & 1.24000 & 1.42400  & 0.76800  & 0.92800 & 0.64000\\
\bottomrule
\end{tabular}
\end{table}


\subsection{Convergence of IMRank}\label{subsection:convergence}

In this section, we first theoretically prove the convergence of IMRank. Then we illustrate the convergence empirically using a real-word network as example.

\begin{theorem}\label{theorem:localswitching}
\textbf{Starting from any initial ranking of nodes, IMRank converges to a self-consistent ranking after a finite number of iterations.}
\end{theorem}

\begin{proof}
We first prove that, with respect to any $k$, the influence spread of the set of top-$k$ nodes, denoted as $I(k)$ for convenience, is nondecreasing in the iterative process of IMRank. After each iteration of IMRank, a ranking $r$ is adjusted to another ranking $r'$. Since IMRank adjusts all nodes in decreasing order of their current ranking-based influence spread $M_r(v)$, the values of $M_r(v_{r'_i})$($1\leq i\leq k$) are the largest $k$ values among all the $M_r(v)$. Hence, there is $I_r(k)=\sum_{1\leq i\leq k}{M_r(v_{r_i})}\leq\sum_{1\leq i\leq k}{M_r(v_{r'_i})}$. Moreover, $I_r(k)=\sum_{1\leq i\leq k}{M_r(v_{r'_i})}$ iff the sets of top-$k$ nodes in ranking $r$ and $r'$ are the same, otherwise $I_r(k)<\sum_{1\leq i\leq k}{M_r(v_{r'_i})}$. Now let's consider a new ranking $r''$ obtained from just reordering the top-$k$ nodes in ranking $r'$ in decreasing order of their ranks in ranking $r$ and keeping the ranks of other nodes still. Apparently, the sets of top-$k$ nodes are the same between ranking $r'$ and $r''$, thus $I_{r'}(k)=I_{r''}(k)$. Then, for each node $v_{r'_i}$, the set of nodes ranked higher than it in ranking $r''$ is definitely a subset of the set of nodes ranked higher than it in ranking $r$. According to the submodularity of influence spread function, we can obtain $M_{r}(v_{r'_i})\leq M_{r''}(v_{r'_i})$ for each node $v_{r'_i}$ ($1\leq i\leq k$). Thus, there is $\sum_{1\leq i\leq k}{M_r(v_{r'_i})}\leq\sum_{1\leq i\leq k}{M_{r''}(v_{r'_i})}=I_{r''}(k)$. Note we have proved $I_r(k)\leq\sum_{1\leq i\leq k}{M_r(v_{r'_i})}$ and $I_{r'}(k)=I_{r''}(k)$. Taken together, we can obtain $I_{r}(k)\leq I_{r'}(k)$, and the equal-sign is tenable iff the sets of the top-$k$ nodes in ranking $r$ and $r'$ are the same, otherwise $I_{r}(k)< I_{r'}(k)$.

Based on the above conclusion, as long as the current ranking is not a self-consistent ranking, in each iteration all the values of $I(k)$($1\leq k\leq n$) are nondecreasing, and at least one $I(k)$ increases. Since $1\leq k\leq n$ and $I(k)$ for each $k$ has an upper bound (i.e., $n$), IMRank eventually converges to a self-consistent ranking within a finite number of iterations, starting from any initial ranking.
\end{proof}

In fact, the above proof also explains the effectiveness of IMRank that it consistently improves the influence spread of top-$k$ nodes for any $k$, resulting in a quick convergence which is much faster than greedy algorithms.
We now empirically illustrate the convergence of IMRank, using a scientific collaboration network, namely HEPT, extracted from the ``High Energy Physics-Theory'' section of the e-print arXiv website arXiv.org. This network is composed of $15K$ nodes and $59K$ edges. We run IMRank to select 50 seed nodes. Figure~\ref{fig:hep:convergence_set_rank} shows the percent of different nodes in two successive iterations. For two widely-used models, weighted independent cascade (WIC) model~\cite{Kempe2003} and trivalency independent cascade (TIC) model~\cite{Chen2010}, the set of top-$50$ nodes becomes unchanged after $5$ and $8$ iterations respectively. Clearly, IMRank converges significantly quicker than greedy algorithms, which requires $k$ iteration for selecting $k$ seed nodes. Figure~\ref{fig:hep:convergence_influencespread} depicts the influence spread of top-$50$ nodes. For convenience, we employ the relative influence spread, i.e., the ratio of the influence spread of top-$50$ nodes in each iteration to the final influence spread obtained when IMRank converges. IMRank only takes $3$ and $5$ iterations to achieve a stable and high influence spread under the two models respectively. The influence spread of top-$k$ nodes always converges with smaller number of iterations than the convergence of the set of top-$k$ nodes. Therefore, one can stop IMRank safely in practice by checking the change of top-$k$ nodes between two successive iterations. 

In sum, we have theoretically and empirically demonstrated the convergence of IMRank. Indeed, the convergence of IMRank could be affected by the estimation of marginal influence spread. Extensive experiments further show IMRank with the LFA strategy always converge quickly in Section~\ref{section:experiment}.

\begin{figure}[t]
\centering
\subfigure[Top-$50$ nodes]
{\label{fig:hep:convergence_set_rank} 
\includegraphics[width=0.48 \linewidth]{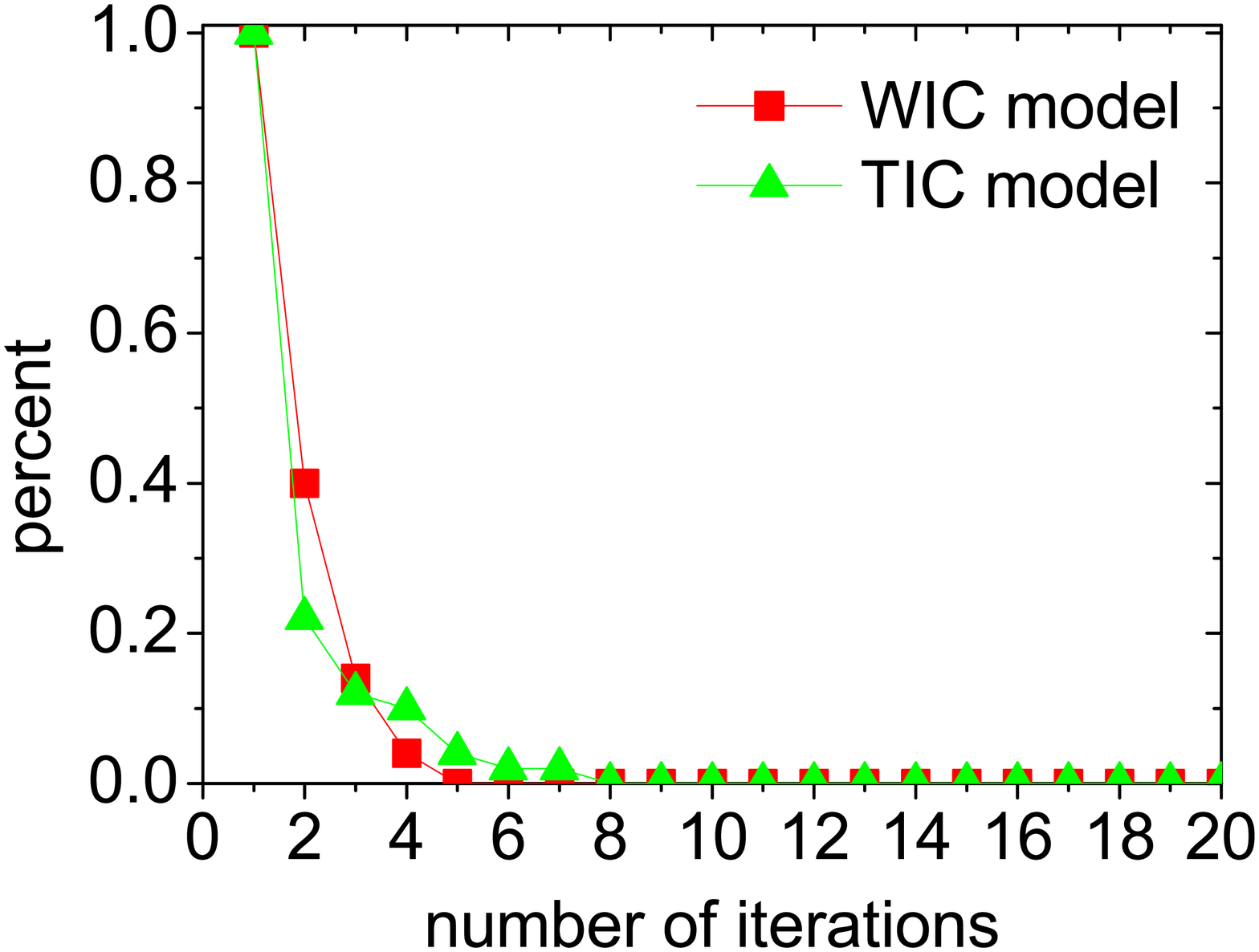}}
\subfigure[Influence spread]
{\label{fig:hep:convergence_influencespread}
\includegraphics[width=0.48 \linewidth]{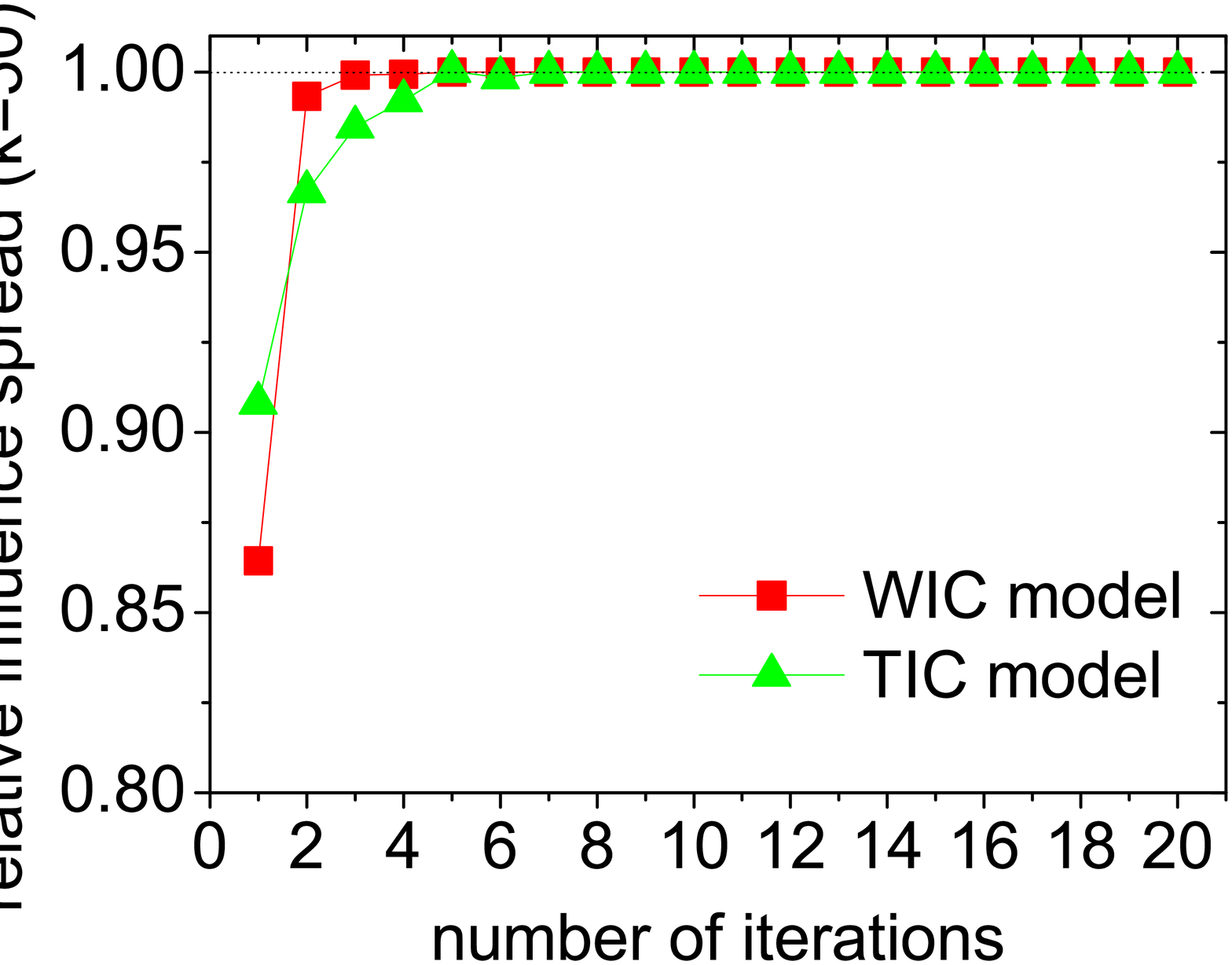}}
\caption{\label{fig:convergence:hep} Convergence of IMRank}
\end{figure}
\subsection{Analysis of initial ranking}\label{subsection:initialranking}

Since IMRank is guaranteed to converge to a self-consistent ranking from any initial ranking, it is necessary to extend the discussion to its dependence on the initial ranking: does an arbitrary initial ranking results in a unique convergence? If not, what initial ranking corresponds to a better result? We explore those questions by empirically simulating IMRank with five typical initial rankings as follows,

\begin{itemize}
  \item \textbf{Random}: Nodes are initially ranked randomly;
  \item \textbf{Degree}: Nodes are initially ranked in descending order of degrees (undirected networks) or out-degrees (directed networks);
  \item \textbf{InversedDegree}: Nodes are initially ranked in ascending order of degrees (undirected networks) or out-degrees (directed networks);
  \item \textbf{Strength}: Nodes are initially ranked in descending order of node strengths (undirected networks) or node out-strengths (directed networks). The node strength is the sum of all weights on its edges. The node out-strength is the sum of all weights on its out-edges;
  \item \textbf{PageRank}: Nodes are initially ranked in descending order of PageRank scores~\cite{Brin1998}, with the default value $0.15$ for the damping factor parameter.
\end{itemize}

Empirical results on the HEPT dataset under the WIC model are reported in Figure~\ref{fig:initialSort}, to compare the performance of IMRank with different initial rankings, as well as the performance of those rankings alone. We also report the performance of classic greedy algorithm for comparison, implemented with CELF optimization~\cite{Leskovec2007}. Performance of IMRank with Random initial ranking, and that of the Random ranking alone, are averaged over $50$ trials.

With the empirical results we conclude:
\begin{itemize}
  \item With different initial rankings, IMRank could converge to different self-consistent rankings. However, IMRank consistently improves the initial rankings in terms of obtained influence spread.

  \item Comparable with the greedy algorithm, IMRank with a ``good'' initial ranking such as Degree, Strength, and PageRank show indistinguishable performance, shown in a single curve in Figure ~\ref{fig:initialSort:ISWIC}. A good initial ranking prefers nodes with high influence;

  \item IMRank with a ``neural'' initial ranking such as random also shows fair performance, slightly poorer than the greedy algorithm and IMRank with a good initial ranking. A neural initial ranking prefers no nodes;

  \item IMRank with a ``bad'' initial ranking such as InversedDegree shows remarkably improvements upon the initial ranking alone but is dominated by the greedy algorithm. A bad initial ranking prefers nodes with low influence.
\end{itemize}
Therefore, IMRank is robust to the selection of initial ranking, and IMRank works well with an initial ranking that prefers nodes with high influence, which could be obtained efficiently in practice. A possible explanation is the priori bias that a high-ranked node earns more allocated influence than a low ranked node, even with the same topological circumstance. Therefore, it helps IMRank to converge to a good ranking if the nodes with high influence are initially ranked high. 

Among the three ``good'' initial rankings with indistinguishable performance, Degree offers a good candidate of initial ranking, since computing the initial ranking consumes a large part in the total running time of IMRank, as shown in Figure~\ref{fig:initialSort:TimeWIC}. 

\begin{figure}[t]
\centering
\subfigure[Influence spread]
{\label{fig:initialSort:ISWIC} 
\includegraphics[width=0.65 \linewidth]{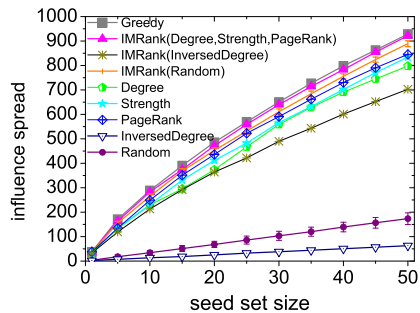}}
\subfigure[Running time when $k=50$]
{\label{fig:initialSort:TimeWIC}
\includegraphics[width=0.65 \linewidth]{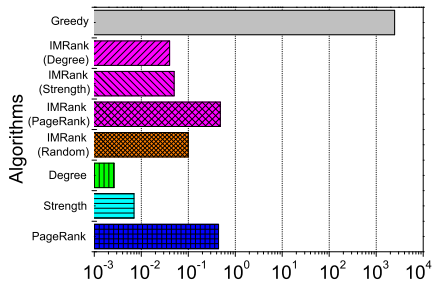}}
\caption{\label{fig:initialSort}Comparison between IMRank with different initial rankings under the WIC model.}
\end{figure} 

\section{Advanced IMRank}\label{section:advance}

In the LFA strategy, a node $v_{r_i}$ is only allowed to allocate its influence to a higher ranked neighboring node $v_{r_j}$, implying the assumption that a node can only be activated by higher ranked neighbors. The assumption ignores the possibility that a lower ranked neighbor $v_{r_j}$ activates a higher ranked node $v_{r_i}$ by playing the role of an intermediate agent of another node $v_{r_i}$ with $k<i$. Take the path $\langle v_1, v_3, v_2\rangle$ in Figure~\ref{fig:last-to-first-strategy} for example. After $v_1$ is selected as a seed, it activates $v_3$ and then $v_3$ as an intermediate agent activates $v_2$.

To combat the above problem, we propose a generalized LFA strategy that trades a slight increase in running time for better accuracy in estimating $M_r$, through exploring more paths that potentially propagate influence. This generalized LFA strategy can further improve the performance of IMRank on influence spread.
In order to avoid duplicate computing that a long path is contained in another longer path, we introduce the \emph{influence paths} as corrections.

\begin{definition}{\textbf{Influence path:}}
Given a ranking $r$, a simple path $d_r(v_{r_j},v_{r_i}) = \langle v_{r_j}, \cdots, v_{r_i}\rangle$ is called an influence path if $v_{r_j}$ is the only node along the path that is ranked higher than $v_{r_i}$.
\end{definition}
\begin{lemma}
A directed edge $\langle v_{r_{j}},v_{r_{i}}\rangle$ is an influence path if $j<i$.
\end{lemma}
\begin{lemma}
A node $v_{r_{i}}$ allocates influence score to another node $v_{r_{j}}$ only along an influence path $d_r(v_{r_{j}},v_{r_{i}})$, if exists any.
\end{lemma}
\begin{proof}
Consider a path $d_r(v_{r_j},v_{r_i})$. If $j>i$, $v_{r_{j}}$ has no chance to trigger a cascade to activate $v_{r_{i}}$, immediately or eventually. Therefore a path is not negligible only when $j<i$. Furthermore, if there is an intermediate node $v_{r_{k}}$ with $k<j$, there is no chance that $v_{r_{j}}$ activates $v_{r_{i}}$ along this path since $v_{r_{k}}$ is triggered earlier, thus such a path can be neglected. If there exists an intermediate node $v_{r_{k}}$ with $j<k<i$, the influence allocated from $v_{r_{k}}$ to $v_{r_{j}}$ already contains the fraction that $v_{r_{j}}$ activates $v_{r_{i}}$, as discussed in Section~\ref{subsection:calculate M}. Thus such a path should not be counted to avoid duplicate computing.
\end{proof}

We denote $\rho_r(v_{r_i},v_{r_j})$ to the probability that $v_{r_i}$ is activated by $v_{r_j}$ through any influence path. $\rho_r(v_{r_i},v_{r_j})$ is equal to the probability that any influence path from $v_{r_j}$ to $v_{r_i}$ has all its nodes activated, discounted by the probability that $v_{r_i}$ is already activated before $v_{r_j}$ attempts. $\rho_r(v_{r_i},v_{r_j})$ can be obtained as follows,
\begin{equation}\label{probabilitywithrank}
\begin{aligned}
\rho_r(v_{r_i},v_{r_j}) = &\left( \prod_{1\leq k<i}  \left(1-\rho_r(v_{r_i},v_{r_k}) \right)\right)\\  &\left(1-\left(\prod_{d_r(v_{r_j},v_{r_i})
\in D_r(v_{r_j},v_{r_i})} (1-p(d_r(v_{r_j},v_{r_i}))\right)\right),
\end{aligned}
\end{equation}
where $p(d_r(v_{r_j},v_{r_i}))=\prod_{\langle v_{r_x}, v_{r_y} \rangle \in  d_r(v_{r_j},v_{r_i})} p_{v_{r_x},v_{r_y}}$ is the joint probability that $v_{r_j}$ activates all nodes on an influence path $d_r(v_{r_j},v_{r_i})$, and $D_r(v_{r_i},v_{r_j})$ denotes the set of all the influence paths starting from $v_{r_j}$ and ending with $v_{r_i}$.
To summarize, the generalized LFA strategy calculates marginal influence spread by replacing the allocation method: a node $v_{r_i}$ delivers a fraction of its influence to each higher-ranked node instead of each adjacent higher-ranked node, with $p(\cdot, v_{r_i})$ replaced by $\rho(\cdot, v_{r_i})$.

Although searching all influence paths takes exhausting computation, we can safely limit the higher-order correction to a second-order or third-order correction to avoid expensive computation. Specifically, we prune paths longer than $l$ hops which are expensive to count but propagate influence with low probabilities. Therefore the marginal influence spread allocation operation is restricted within a local region, avoiding exploring the whole network. Obviously $l=1$ makes the generalized LFA strategy collapsed into the LFA strategy.




The time and space complexity of IMRank with the generalized LFA strategy mainly depends on $l$. Let $d_{max}$ denote to the largest number of paths with length of $l$ ends in an arbitrary node. The time required for scanning any node is $O(d_{max} \log d_{max})$, used for searching candidate nodes, sorting candidate nodes by their ranks, and allocating influence. Hence the total time complexity of IMRank is $O(n T d_{max}\log d_{max})$, where $T$ is the number of iterations needed for the convergence of IMRank. Our experiment results show that, IMRank always converges with a fairly small $T$ significantly smaller than $k$, e.g, $T<10$ when $k=50$. Since $d_{max}$ is usually much smaller than $n$, e.g. $d_{max}$ is just the largest indegree among all nodes when $l=1$, the time complexity of IMRank is low. Talking about the space (memory) complexity, IMRank only needs to store the value of $M_r(v)$ for each node, which takes $O(n)$ space in memory. Hence the space complexity is also low. 

Figure~\ref{fig:l:performance} shows the impact of $l$ on the performance of IMRank, measured on the NEPT network with the WIC model and $k=50$ for example. We compare the results of IMRank with Degree and Random initial rankings since the results for other initial rankings are similar. It shows that, when $l$ increases from $1$ to $2$, there is a visible increase on the performance of IMRank, measured with influence spread. It indicates that, a larger $l$ indeed makes the estimation of marginal influence spread more accurate, and further makes IMRank obtain better ranking. When $l$ increases beyond $2$, the performance of IMRank converges fast, because the propagation probabilities of long paths decrease exponentially with the length. Hence, long influence paths impact little on the final estimation. As shown in the inset figure of Figure~\ref{fig:l:performance}, the running time of IMRank increases rapidly as $l$ increases, since much more paths need searching. Balancing the trade-off between the influence spread and running time of IMRank, a suitable $l$ can be selected based on the practical requirement on accuracy and affordable computational resource.

\begin{figure}[t]
\centering
\includegraphics[width=0.55 \linewidth]{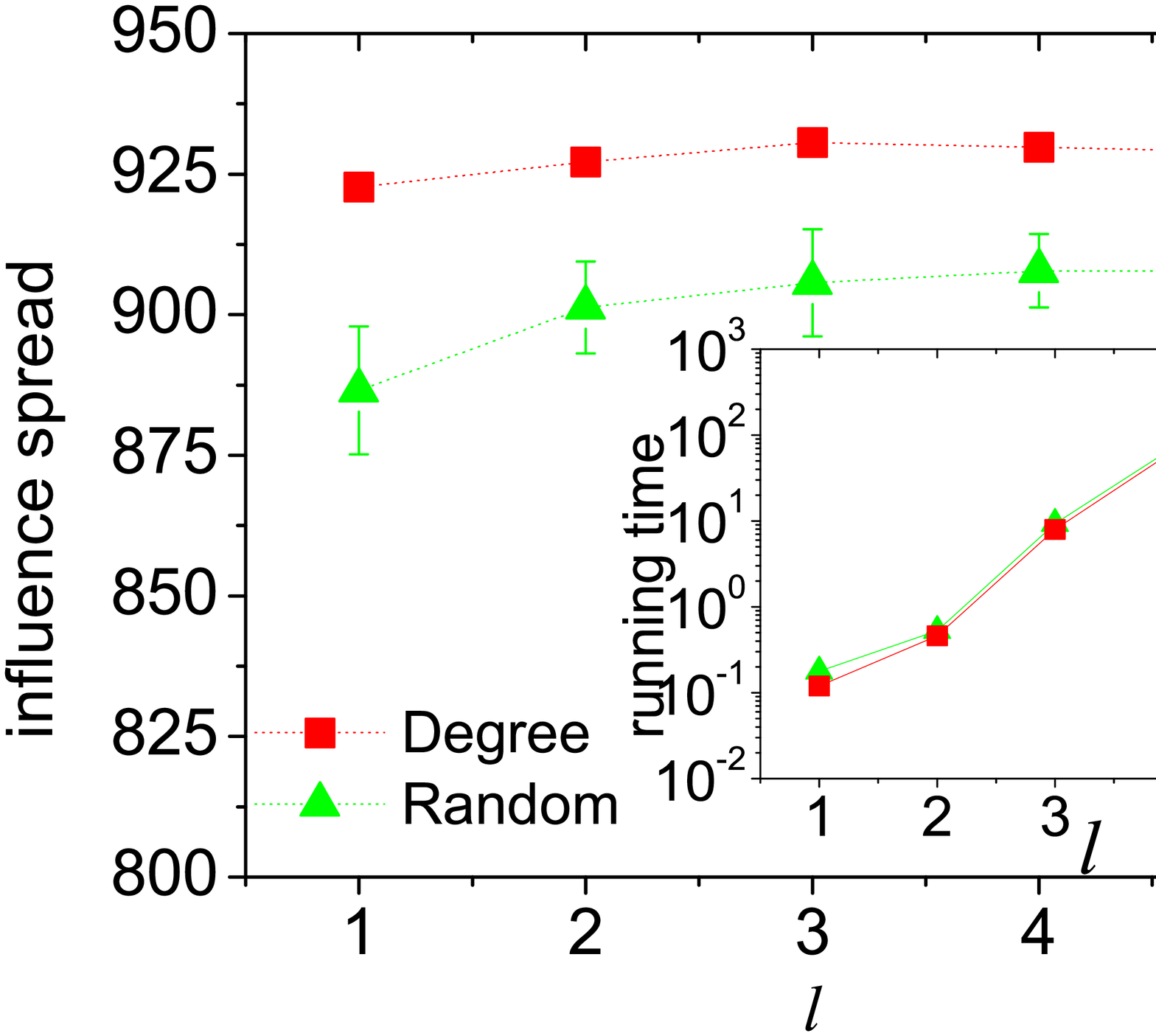}
\caption{Impact of $l$ on the performance of IMRank.}
\label{fig:l:performance}
\end{figure} 

\section{Experiments}\label{section:experiment}

In this section, we evaluate IMRank on real-world networks by comparing IMRank with state-of-the-art influence maximization algorithms.
\begin{figure*}
\centering
\subfigure[WIC model]
{\label{fig:PHY:WICIS} 
\includegraphics[width=0.3 \linewidth]{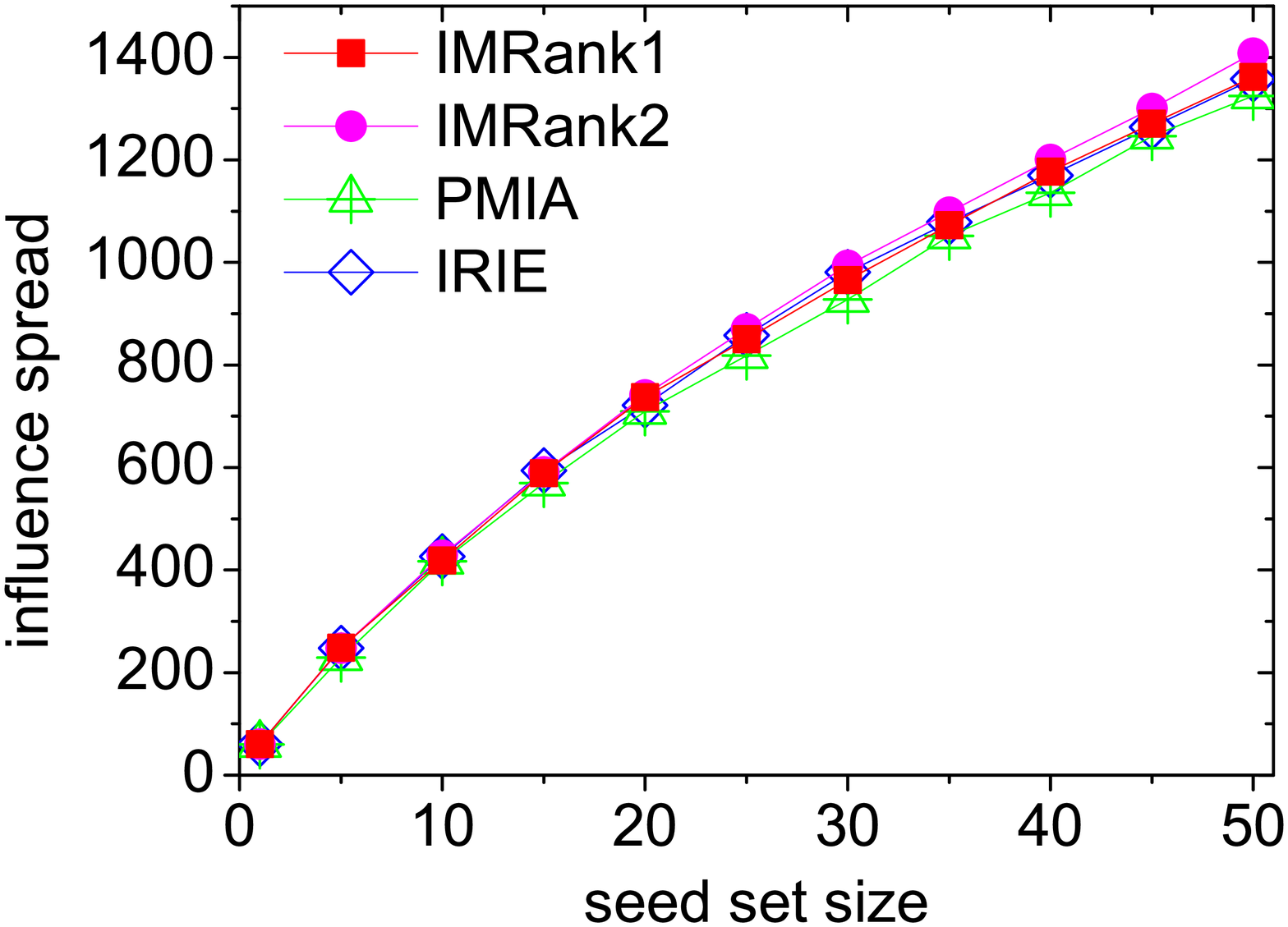}}
\subfigure[TIC model]
{\label{fig:PHY:TICIS}
\includegraphics[width=0.3 \linewidth]{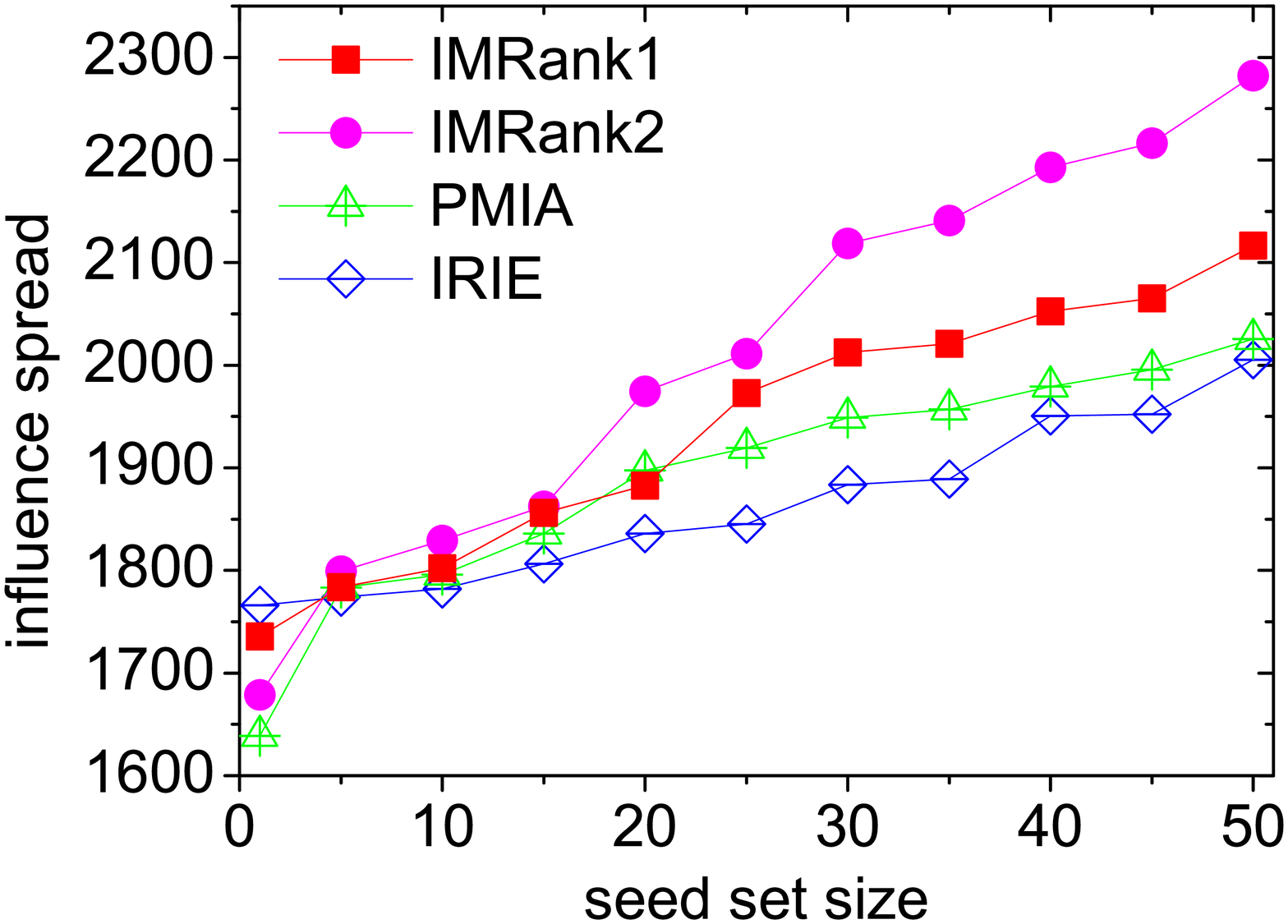}}
\subfigure[Running Time]
{\label{fig:PHY:TIME}
\includegraphics[width=0.3 \linewidth]{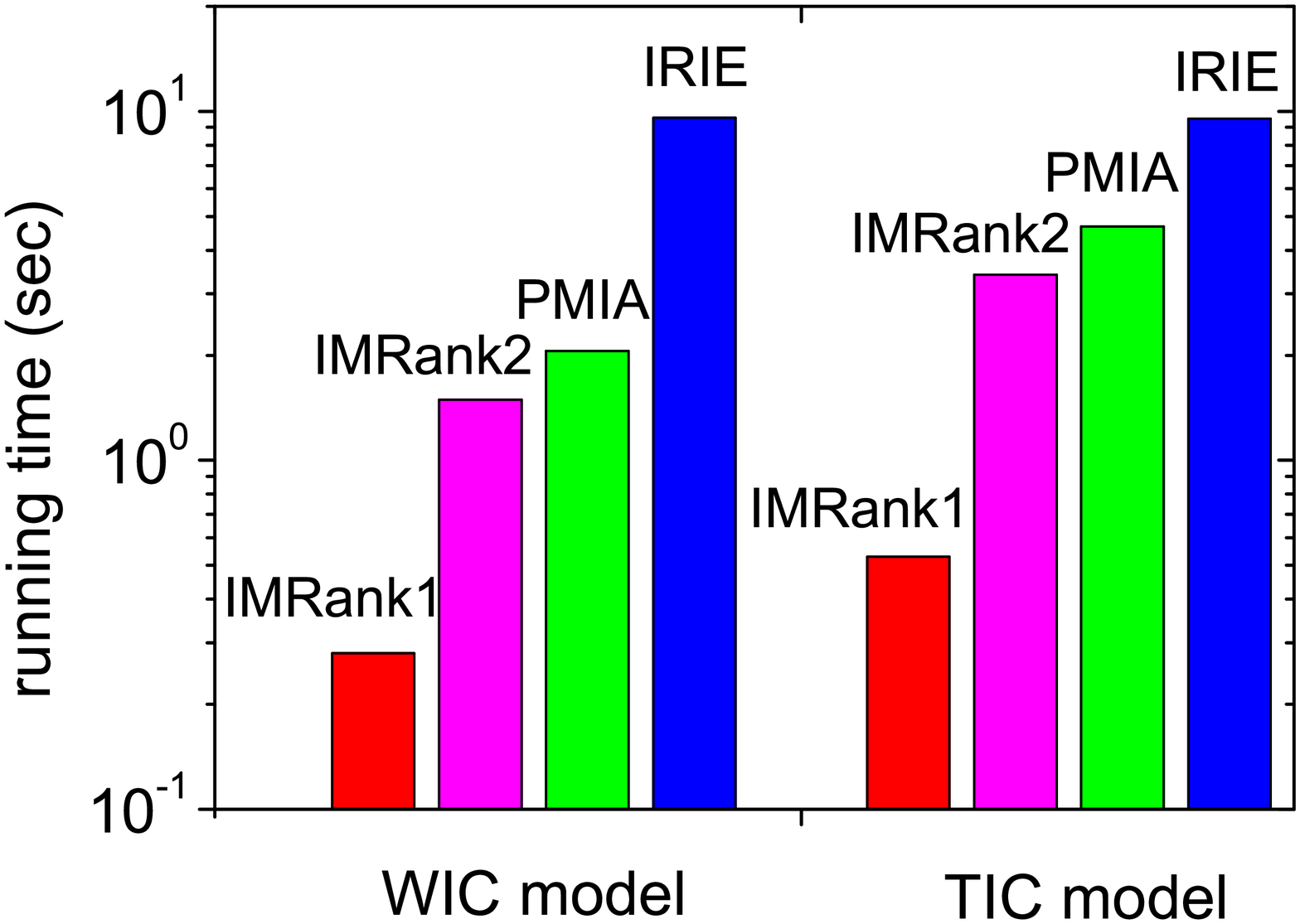}}
\caption{\label{fig:PHY} Influence spread and running time on the PHY dataset}
\end{figure*}

\begin{figure*}[t]
\centering
\subfigure[WIC model]
{\label{fig:DBLP:WIC} 
\includegraphics[width=0.3 \linewidth]{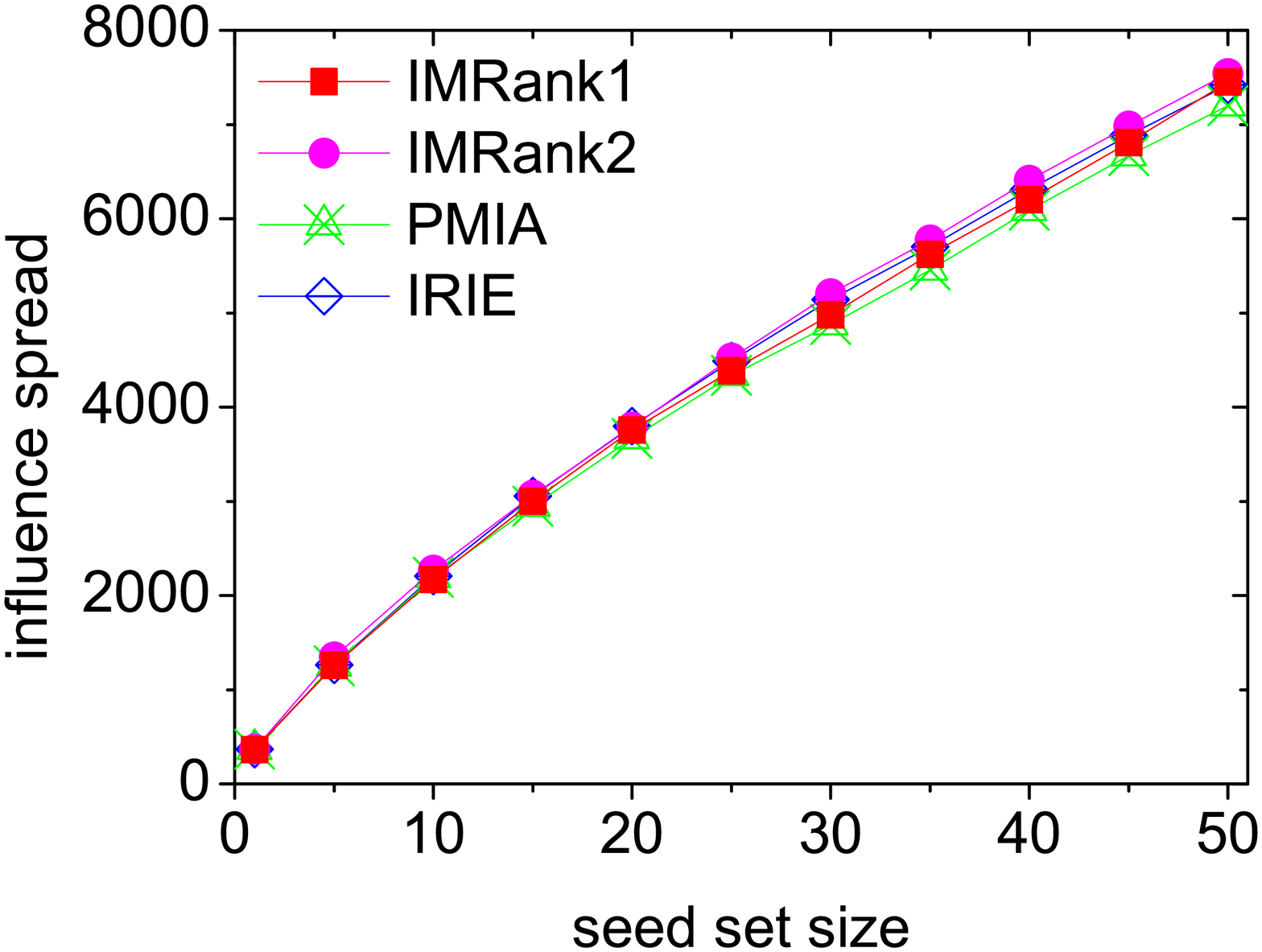}}
\subfigure[TIC model]
{\label{fig:DBLP:TIC}
\includegraphics[width=0.3 \linewidth]{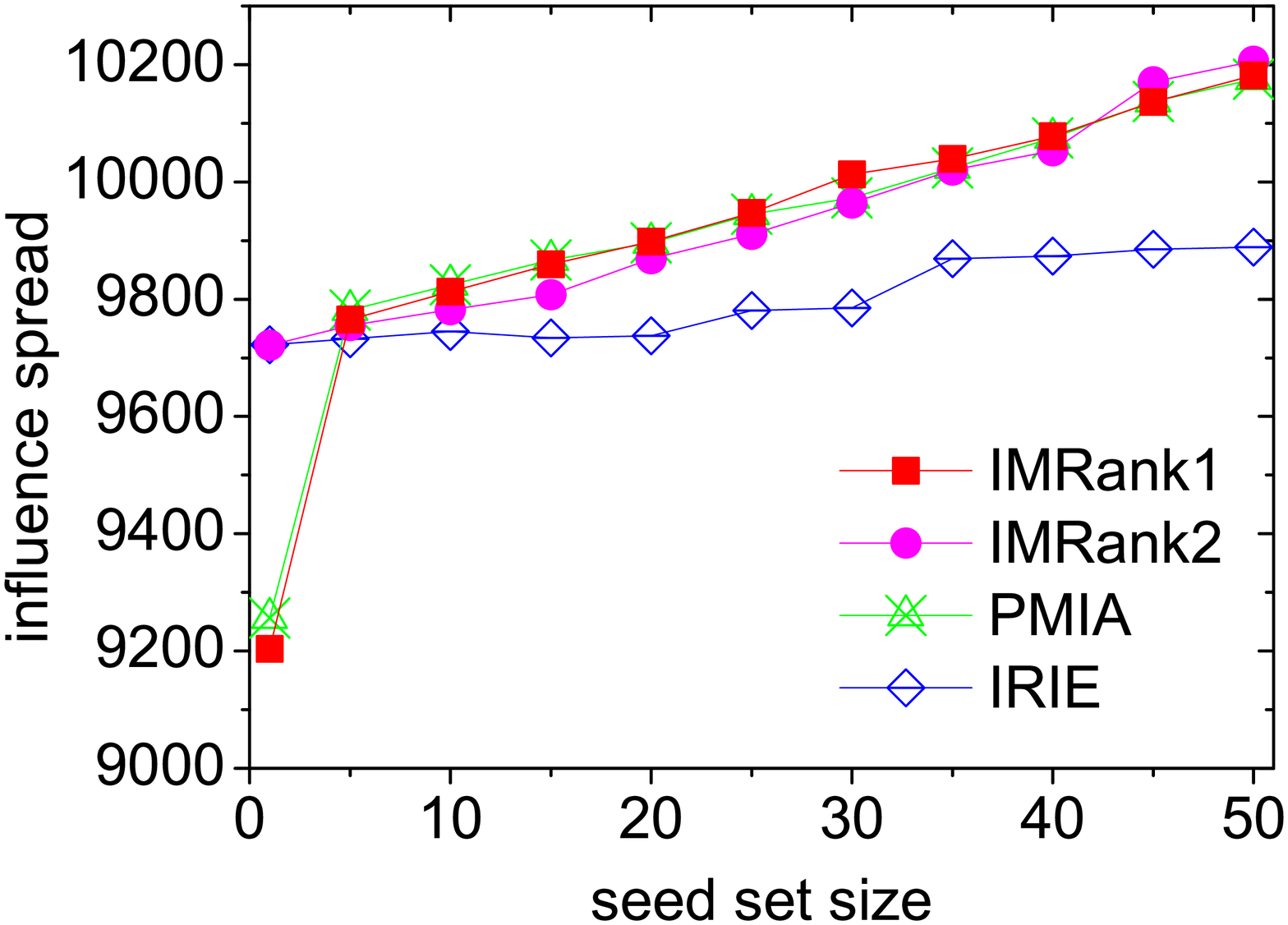}}
\subfigure[Running Time]
{\label{fig:DBLP:TIME}
\includegraphics[width=0.3 \linewidth]{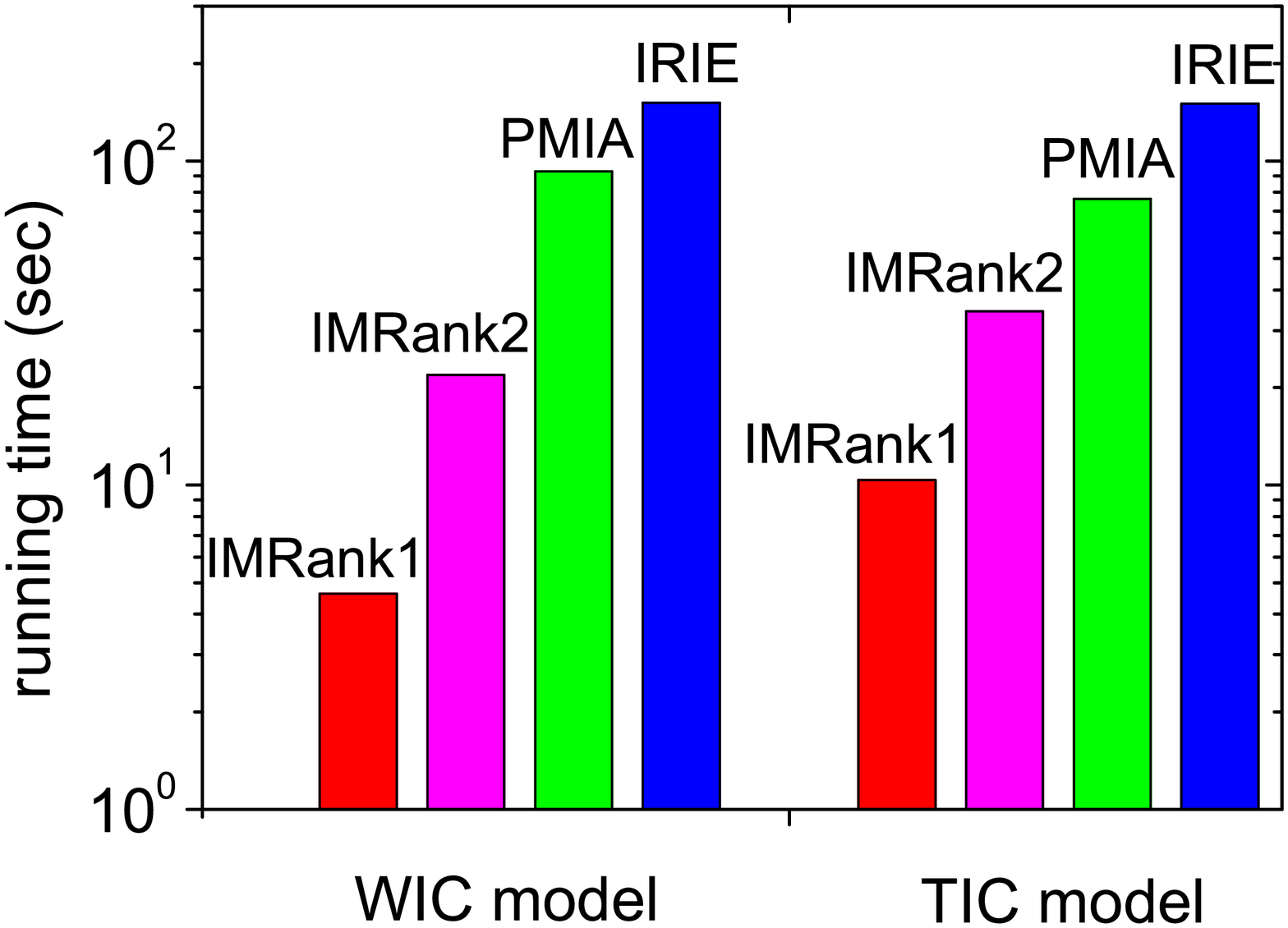}}
\caption{\label{fig:DBLP} Influence spread and running time on the DBLP dataset}
\end{figure*}

\subsection{Experimental Setup}
\subsubsection{Diffusion models}
Experiments are conducted under two widely-used independent cascade models:
\begin{itemize}
  \item \textbf{Weighted independent cascade (WIC) model }~\cite{Kempe2003}: Each edge $(u, v)$ is assigned a propagation probability $p(u,v)=1/d_v$, where $d_v$ is the indegree of node $v$.
  \item \textbf{Trivalency independent cascade (TIC) model }~\cite{Chen2010}: Each edge is assigned a propagation probability selected from \{0.1,0.01,0.001\} in a uniform random manner, indicating high, medium and low levels of influence.
\end{itemize}

\subsubsection{Baseline algorithms} The compared algorithms include two implementations of IMRank and two state-of-the-art heuristic algorithms, i.e., PMIA and IRIE. Details are as follows:
\begin{itemize}
\item \textbf{IMRank1}:
This is the IMRank with Degree as initial ranking method and $l=1$. According to the analysis of section~\ref{subsection:convergence}, we set its stopping criteria as when the sets of top-$k$ nodes are the same during two successive iterations or the iteration runs 10 rounds.
\item \textbf{IMRank2}:
This is the IMRank with Degree as initial ranking method and $l=2$, with the same stopping criteria to IMRank1.
\item \textbf{PMIA}:
This heuristic algorithm estimates influence spread based on maximum influence paths~\cite{Chen2010}. We use the recommended parameter setting $\theta=1/320$.
\item \textbf{IRIE}:
This heuristic algorithm integrates influence ranking with influence estimation~\cite{Jung2012}. The parameters $\alpha$ and $\theta$ are set to be $0.7$ and $1/320$, and the maximum times of iterations for initial round and subsequent rounds are respectively 20 and 5 as recommended.
\end{itemize}

\begin{table}[t]
\centering \caption{Statistics of test networks} \label{table:statisticsoftestnetworks}
\begin{tabular}{p{1.5cm}rrc}
\toprule
Datasets          & \#Nodes    & \#Edges     & Directed? \\
\midrule
PHY               & 37K        & 231K        &undirected \\
DBLP              & 655K       & 2M        &undirected \\
EPINIONS          & 76K        & 509K      & directed  \\
DOUBAN            & 552K       & 22M        &directed \\
LIVEJOURNAL       & 4M         & 69M        &directed \\
\bottomrule
\end{tabular}
\end{table}

\subsubsection{Datasets}
Experiments are conducted on five real-world networks, two undirected scientific collaboration networks and three directed online social networks. Table~\ref{table:statisticsoftestnetworks} gives basic statistics of those networks. One of the two scientific collaboration networks, denoted as PHY, is obtained from the complete list of papers of the Physics section of the e-print arXiv website. The other one, denoted as DBLP, is extracted from the DBLP Computer Science Bibliography~\footnote{http://www.informatik.uni-trier.de/$\sim$ley/db/}. The three online social networks are EPINIONS, DOUBAN, and LIVEJOURNAL~\footnote{EPINIONS and LIVEJOURNAL can be downloaded from http://snap.stanford.edu/data/. DOUBAN can be obtained on demand via email to the authors.}, respectively extracted from the websites of epinions.com, douban.com and livejournal.com. In the EPINIONS dataset, an edge between two users $u$ and $v$, denoted as $\langle u,v\rangle$, represents that user $u$ trusts user $v$. In the DOUBAN dataset~\cite{Huang2012},
an edge between two users $u$ and $v$ represents that user $u$ follows user $v$. In the LIVEJOURNAL network~\cite{Backstrom2006},
an edge between two users $u$ and $v$ represents that user $u$ declares user $v$ as his/her friend. We choose these five networks based on the consideration that these networks possess various kinds of relationships and different sizes ranging from hundreds of thousands edges to millions of edges. Actually we test our method on many other networks. Limited by space, results on these networks are not included in this paper.

All experiments are conducted on a server with 1.9GHz Quad-Core AMD Opteron(tm) Processor 8347HEx4 and 64G memory.

\begin{figure*}[t]
\centering
\subfigure[WIC model]
{\label{fig:EPINIONS:WIC} 
\includegraphics[width=0.3 \linewidth]{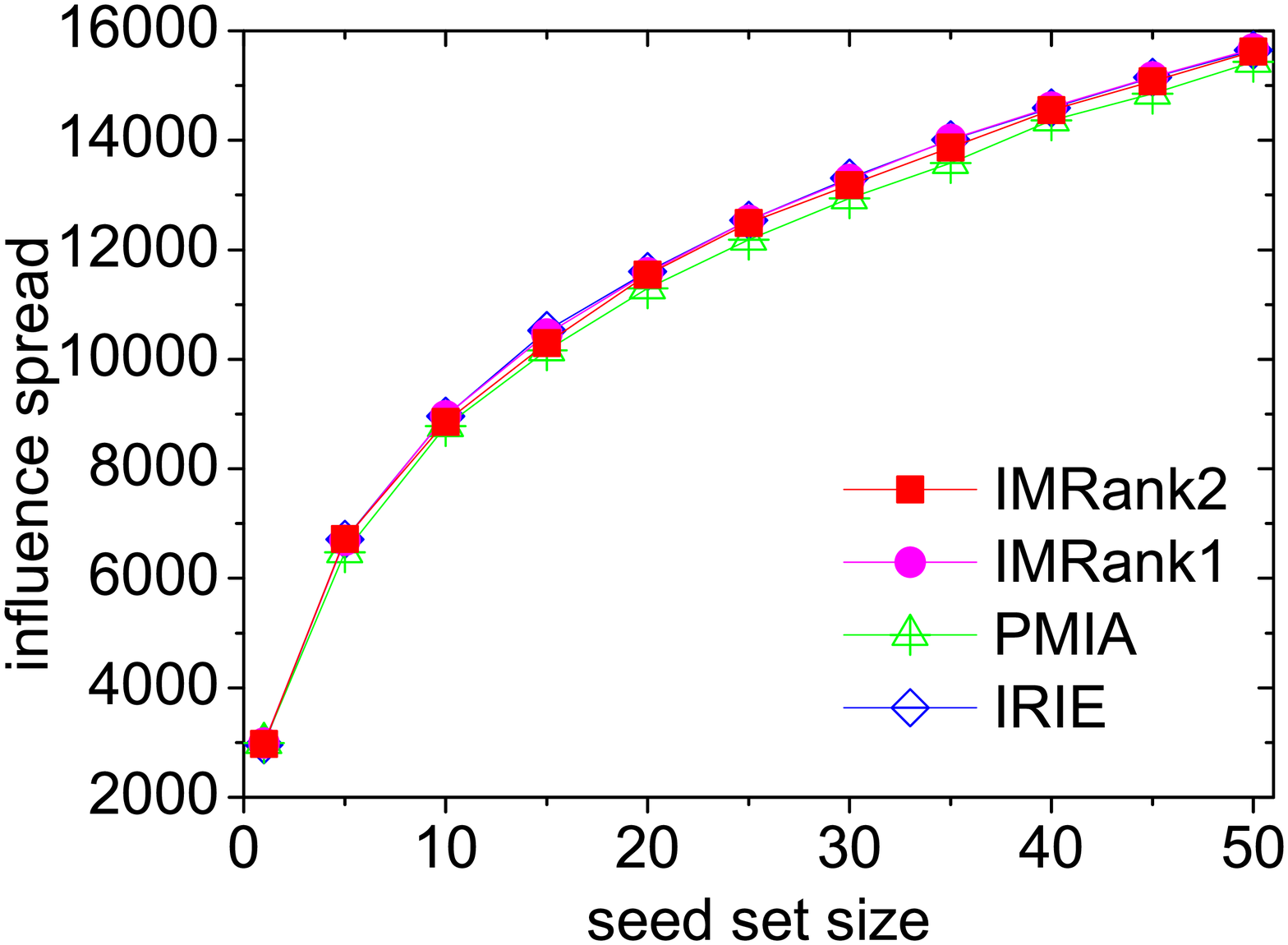}}
\subfigure[TIC model]
{\label{fig:EPINIONS:TIC}
\includegraphics[width=0.3 \linewidth]{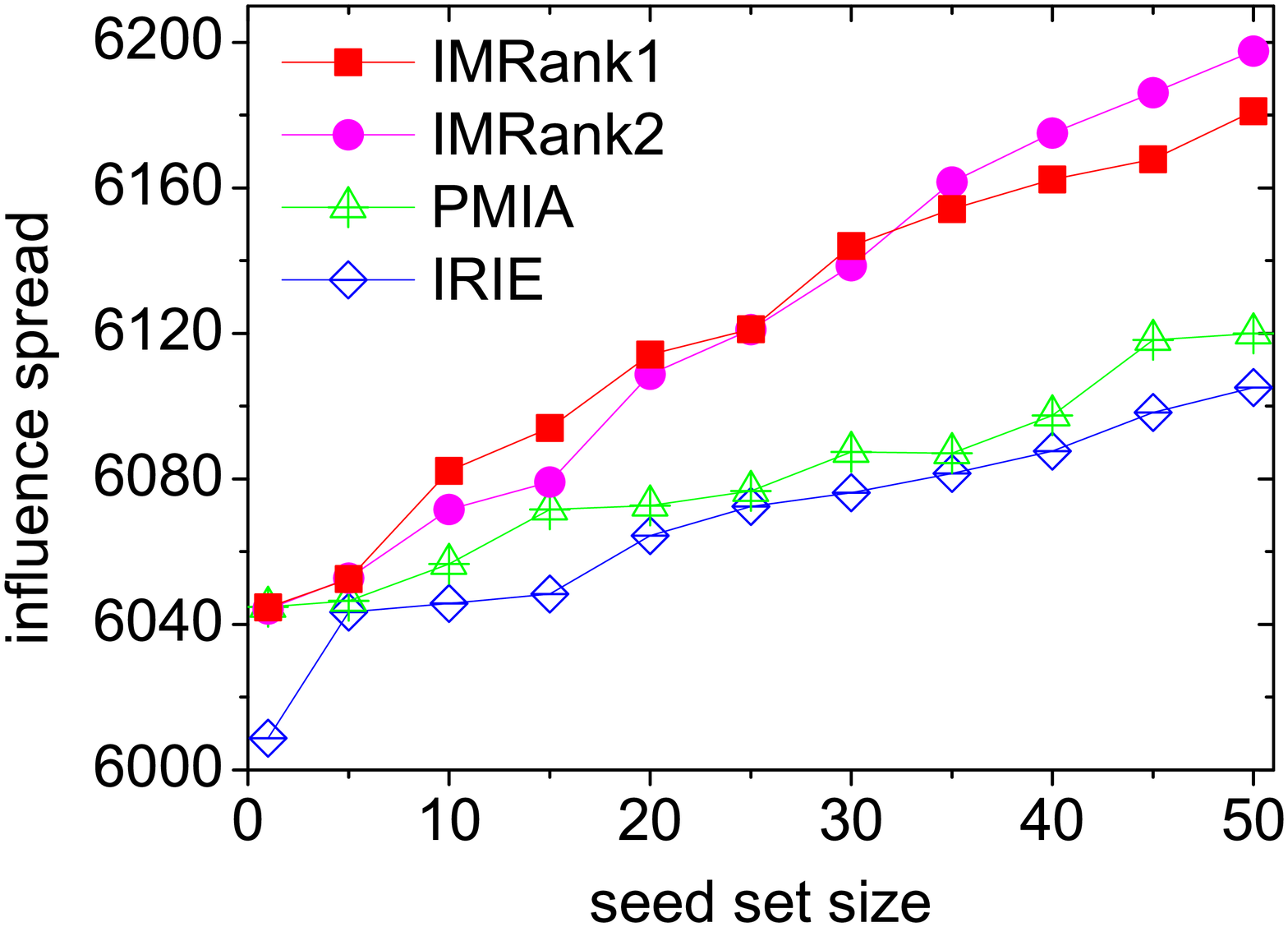}}
\subfigure[Running Time]
{\label{fig:EPINIONS:TIME}
\includegraphics[width=0.3 \linewidth]{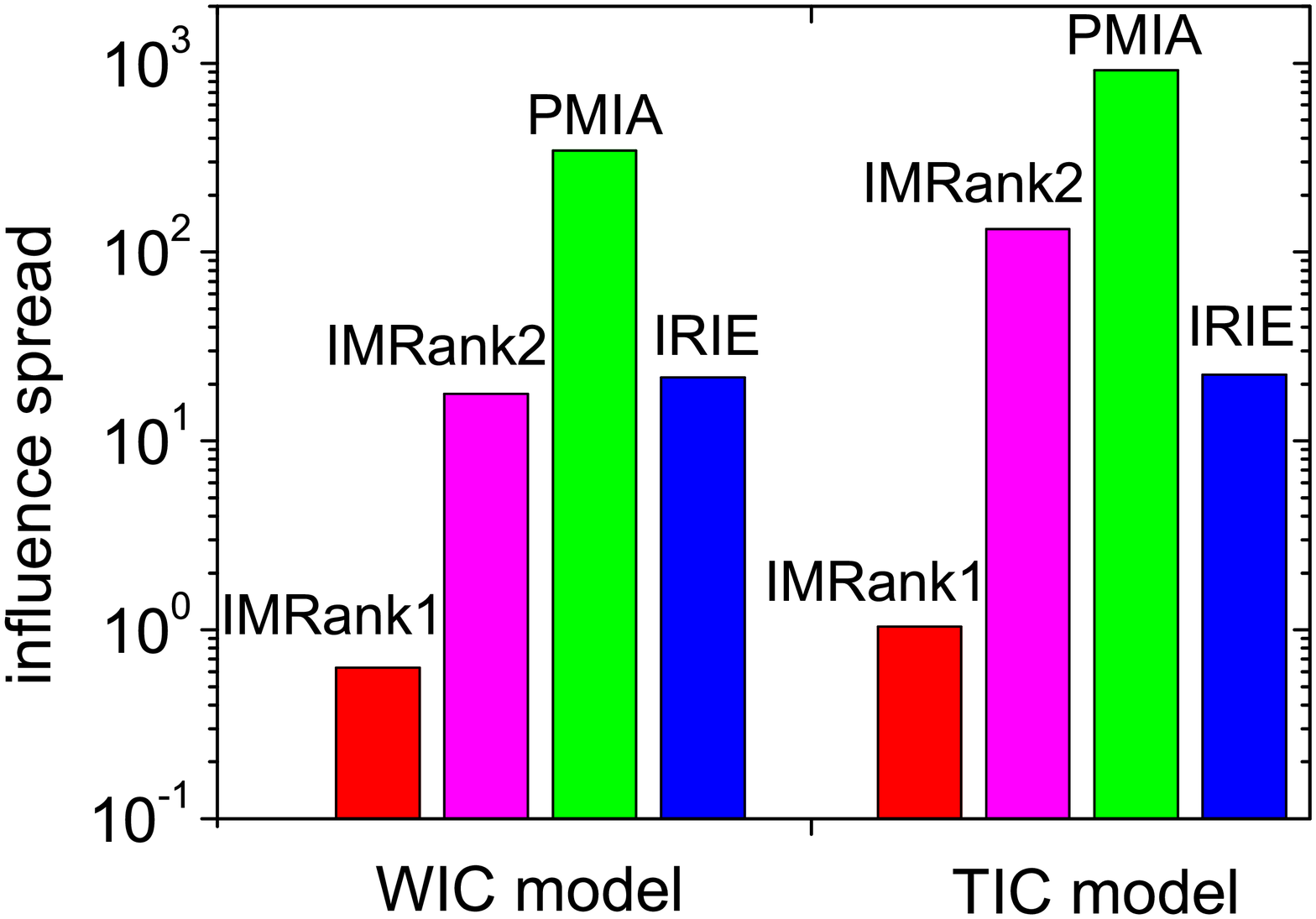}}
\caption{\label{fig:EPINIONS} Influence spread and running time on the EPINIONS dataset}
\end{figure*}

\begin{figure*}[t]
\centering
\subfigure[DOUBAN]
{\label{fig:DOUBANLJ:WIC} 
\includegraphics[width=0.3 \linewidth]{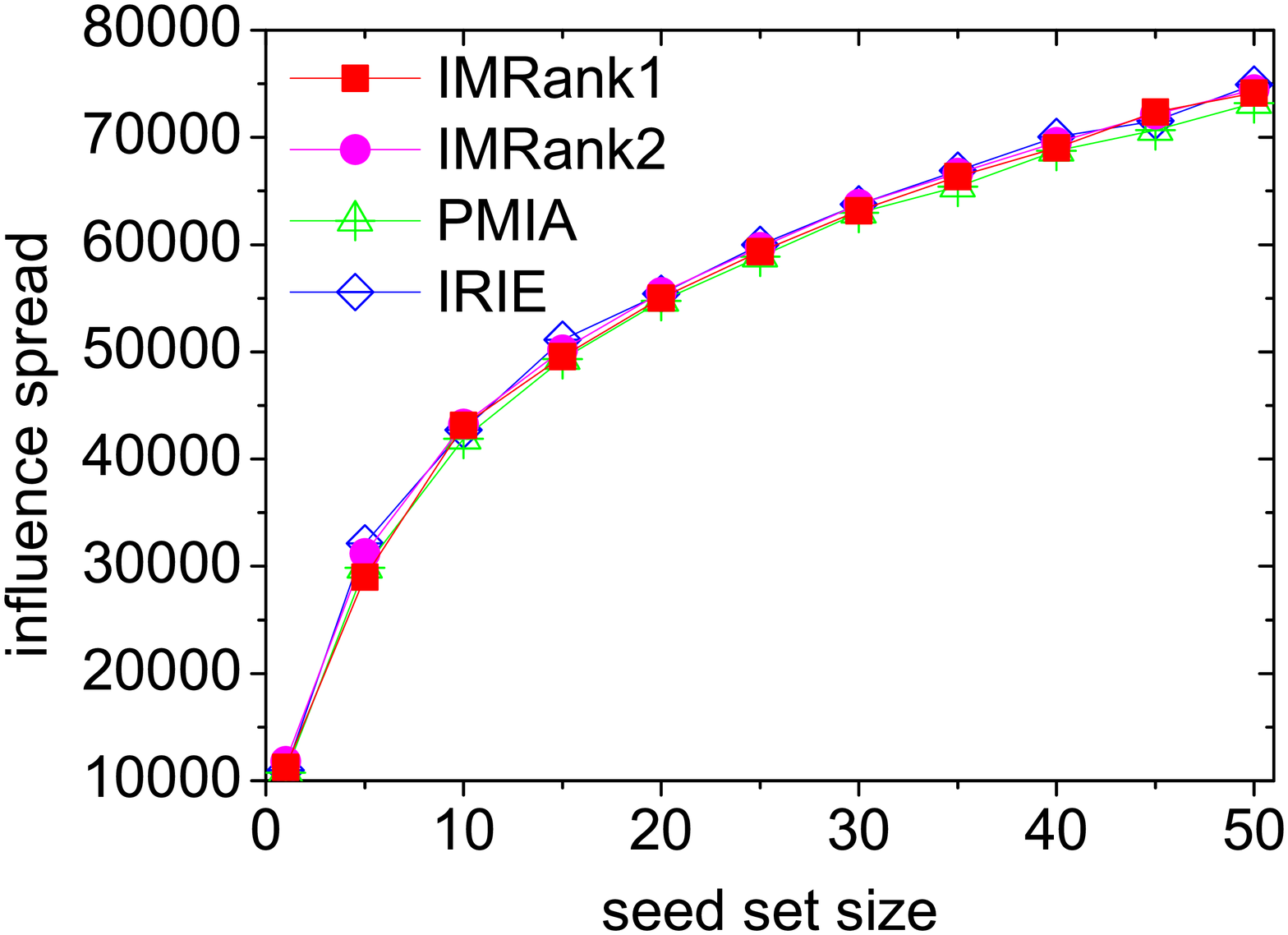}}
\subfigure[LIVEJOURNAL]
{\label{fig:DOUBANLJ:WIC}
\includegraphics[width=0.3 \linewidth]{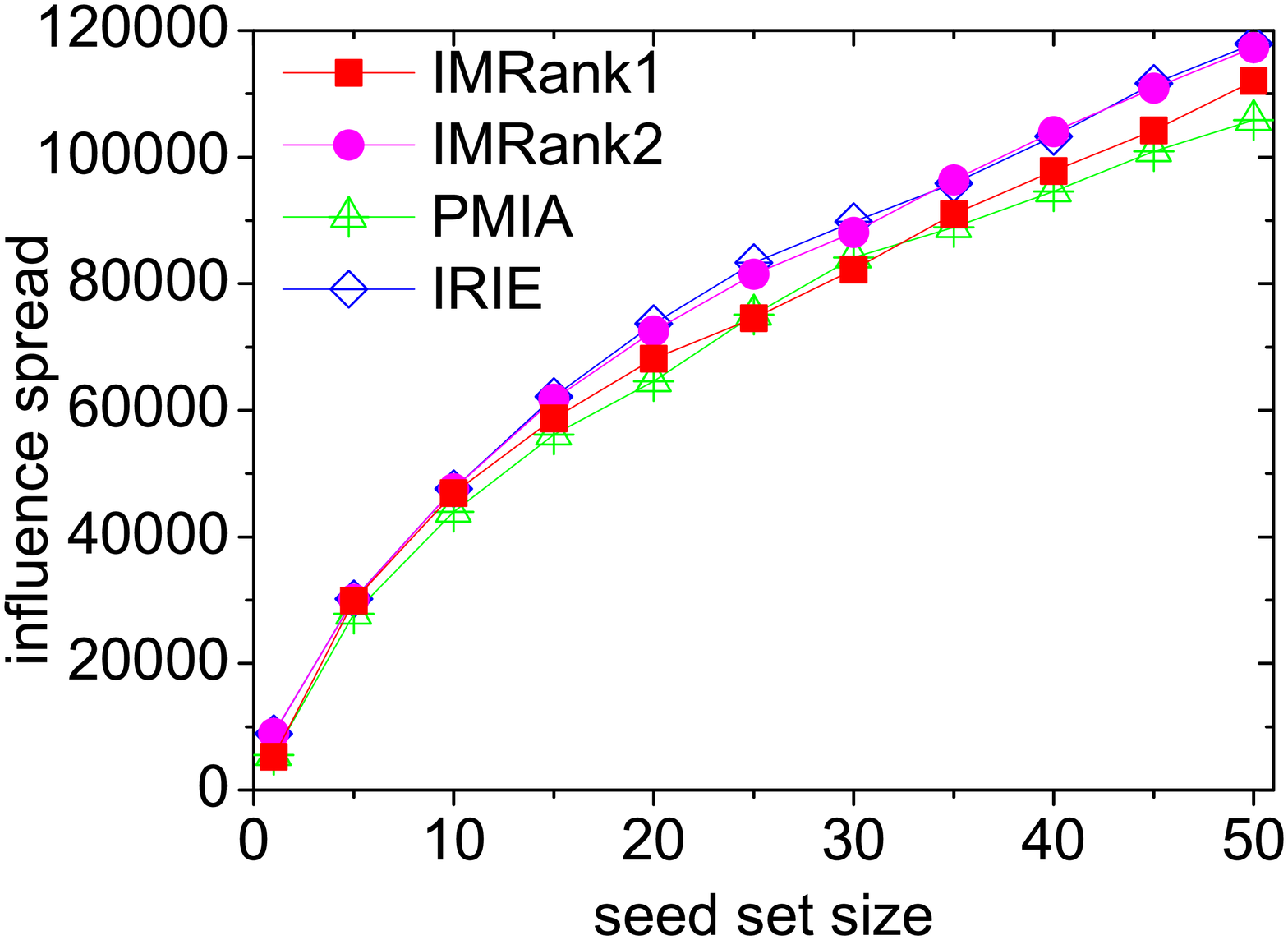}}
\subfigure[Running Time]
{\label{fig:DOUBANLJ:TIME}
\includegraphics[width=0.3 \linewidth]{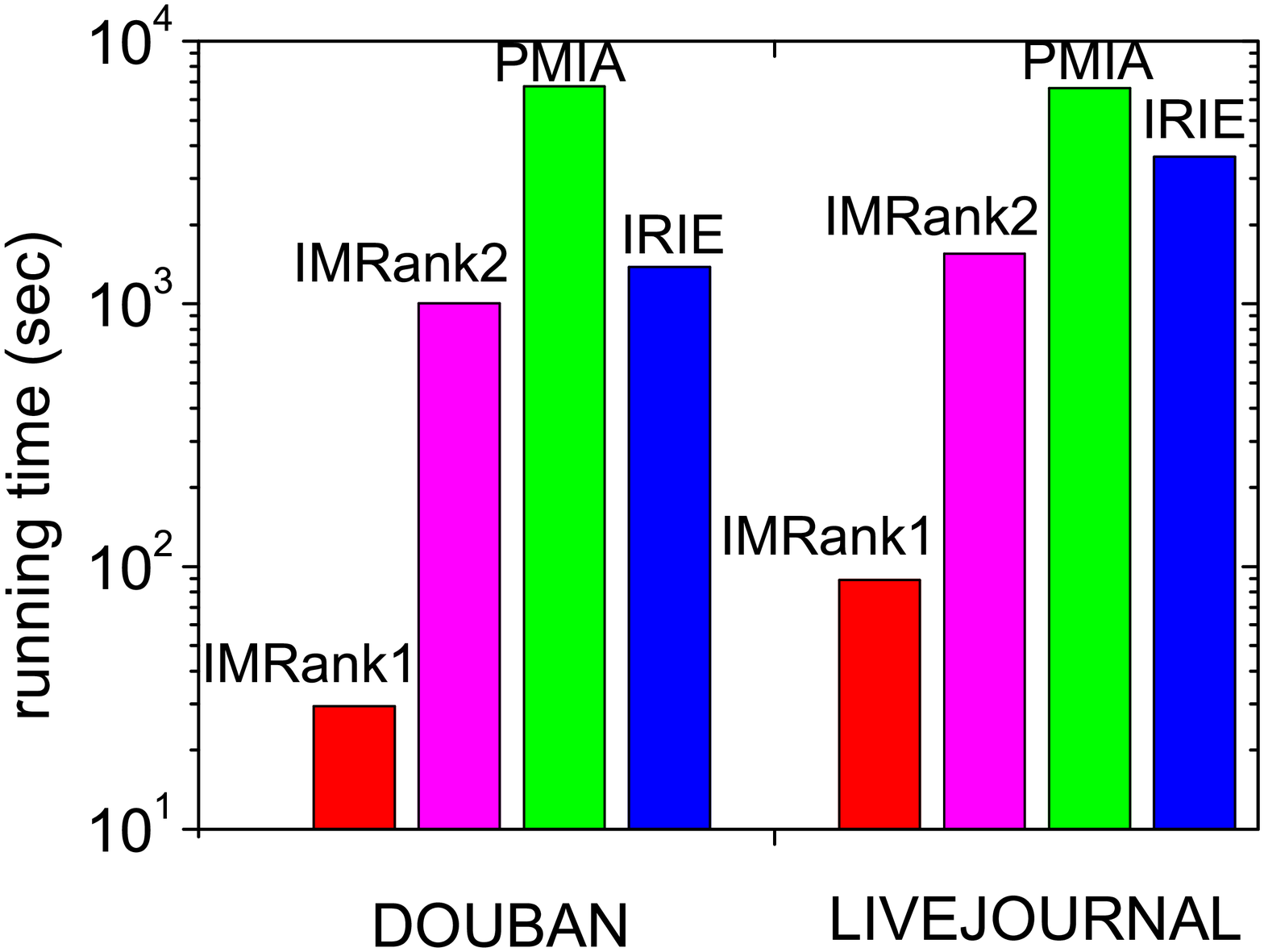}}
\caption{\label{fig:DOUBANLJ} Influence spread and running time on the DOUBAN and LIVEJOURNAL datasets}
\end{figure*}

\subsection{Experimental results}
We evaluate IMRank on real-world networks by comparing it with state-of-the-art algorithms. Evaluation metrics include influence spread and running time. For the comparison of obtained influence spread, we test the cases of $k=1,5,10,15,20,25,30,35,40,45,50$. For the comparison of running time, we focus on the typical case $k=50$. Each figure of Figures~\ref{fig:PHY}-\ref{fig:DOUBANLJ} shows the results on a certain network. The first two subfigures give the results of influence spread under the WIC model and the TIC model respectively, and the last one gives the results of running time. 

Figure~\ref{fig:PHY} shows the experimental results on the PHY dataset. Under the WIC model, IMRank2 achieves the best influence spread, followed by IMRank1, outperforming PMIA and IRIE. The distinguished accuracy of IMRank2 is attributed to the fact IMRank2 explores more influence paths to accurately estimate ranking-based marginal influence spread. PMIA exhibits the worst performance, 6.3\% lower influence spread than IMRank2 when $k=50$. Under the TIC model, as shown in Figure~\ref{fig:PHY:TICIS}, similar results are obtained and the gaps between those algorithms become more visible. For influence spread, IMRank2 and IMRank1 are the top two algorithms while PMIA slightly outperforms IRIE. The influence spread obtained by IMRank2 is 13.8\% and 12.7\% higher than that obtained by IRIE and PMIA respectively. Moreover, as shown in Figure~\ref{fig:PHY:TIME}, IMRank1 and IMRank2 run faster than the competing algorithms under both WIC model and TIC model. IMRank1 is the fastest one followed by IMRank2 which achieves higher influence spread at the cost of longer running time, while PMIA takes the third place and IRIE runs slowest. In particular, the running times of IRIE and PMIA are 30 times and 10 times longer than the running time of IMRank1 under the WIC model respectively, and 18 times and 9 times longer than that of IMRank1 under the TIC model. With the running time dramatically reduced, IMRank1 still achieves better influence spread which is about 5.5\% and 4.5\% higher than that of IRIE and PMIA respectively. The consistent performance of IMRank1 and IMRank2 demonstrates the effectiveness of IMRank. The inconsistent performance of PMIA and IRIE under the two diffusion models illustrates that both PMIA and IRIE are unstable.

Figure~\ref{fig:DBLP} shows the results on DBLP, a network with two millions edges. The performance of the four algorithms on this network is similar to their performance on PHY dataset. For the WIC model, IMRank2 achieves the highest influence spread and IMRank1 is the fastest one. In particular, when $k=50$, the highest influence spread is achieved by IMRank2 and its running time is less than PMIA and IRIE. IMRank1 obtains similar influence spread to PMIA and its running time is one order of magnitude smaller than that of PMIA. For the TIC model, IMRank1, IMRank2 and PMIA achieve very similar influence spread, which is significantly higher than the influence spread achieved by IRIE. Moreover, IMRank1 runs nearly 8 times and 13 times faster than PMIA and IRIE.

Figure~\ref{fig:EPINIONS} gives the results on EPINIONS, a social network with more than half a million edges. For the WIC model, IMRank1 and IMRank2 run faster than PMIA and IRIE. In particular, comparative to PMIA, IMRank1 reduces the running time in more than two orders of magnitudes and IMRank2 reduces the running time in more than one order of magnitude. For the TIC model, IMRank2 achieves the best influence spread and IMRank1 takes the second place. Both IMRank1 and IMRank2 significantly outperform PMIA and IRIE. Moreover, the running time of IMRank1 is only 0.1\% of the running time of PMIA and 5\% of that of IRIE. With similar running time, IMRank2 achieves significant higher influence spread than that of PMIA and IRIE.


Figure~\ref{fig:DOUBANLJ} shows the results on the DOUBAN and LIVEJOURNAL datasets. The number of edges of DOUBAN and LIVEJOURNAL is $22$ millions and $69$ millions respectively. Here we only give the results under the WIC model. On the DOUBAN network, the four algorithms achieve comparable influence spread. However, IMRank1 runs more than two orders of magnitude faster than PMIA and more than one order of magnitude faster than IRIE. On the LIVEJOURNAL network, IMRank2 and IRIE have similar influence spread, while IMRank1 follows and PMIA achieves the lowest influence spread. Note that IMRank2 runs faster than IRIE, and IMRank1 runs much faster than PMIA. We do not show the results under the TIC model since no visible difference is observed among the four tested algorithms. This is due to the fact that selecting one influential node always achieves a very large influence spread on DOUBAN and LIVEJOURNAL networks, and no increase of influence spread can be gained by adding a new seed. Such phenomenon has been observed and discussed in~\cite{Kempe2003} and~\cite{Chen2010}. The possible reason is that the influence networks generated by the TIC model on the two networks have a relatively large strongly connected component. In addition, IRIE runs faster than PMIA on EPINIONS while PMIA runs faster than IRIE on the two scientific collaboration networks, PHY and DBLP. This demonstrates that the both PMIA and IRIE perform unstable on different networks.

%

These experiments clearly show that PMIA and IRIE perform unstable on different scenarios while IMRank consistently shows good performance. According to these experiments, PMIA always runs the slowest among the four tested algorithms on denser networks, such as EPINIONS, DOUBAN and LIVEJOURNAL. This is mainly because such networks involve lots of influence paths to calculate and store. In contrast, IRIE always performs the worst on sparser and smaller networks, PHY and DBLP. This is probably because IRIE strictly obeys the iterative ranking and iterative estimation, resulting in relatively long time in sparser and smaller networks. Different from the two algorithms, IMRank seems to perform efficient and stable among different tested cases. IMRank1 always runs more than one order of magnitude faster than PMIA and IRIE when they achieve similar influence spread. IMRank2 consistently provides better influence spread than PMIA and IRIE, but runs faster than them.

\section{Conclusions}\label{section:conclusion}

In this paper, we investigated influence maximization from a novel ranking perspective. We proposed an efficient iterative framework IMRank to explore the benefits of accurate greedy algorithms and efficient heuristic estimation of influence spread. This framework effectively tunes any initial ranking into a self-consistent ranking in an iterative manner through fully leveraging the interplay between the ranking of nodes and their ranking-based marginal influence spread. A last-to-first allocating strategy is further proposed to efficiently estimate the ranking-based marginal influence spread. This strategy is elaborately designed according to the characteristics of the independent cascade model and the ranking-based marginal influence spread. We further generalize the last-to-first allocating strategy in order to achieve more accurate estimation. We also prove the convergence of IMRank and analyze the impact of initial ranking. Moreover, IMRank always work well with simple heuristic rankings, such as degree, strength. Extensive experiments on large scale real-world social networks demonstrate the efficiency of IMRank. Its scalability outperforms the state-of-the-art heuristics while its accuracy is comparable to the greedy algorithms.

For future work, we will try to analyze the accuracy of IMRank theoretically. Moreover, we believe our proposed iterative framework is of generality for the some cases greedy algorithm is suitable for. We will try to extend it to other problems beyond influence maximization, such as diversity problem in retrieval.

\end{document}